\documentclass[a4paper,11pt]{article}
\usepackage{amsmath,amssymb,amscd,amsthm}
\usepackage{a4wide}
\usepackage{tikz}

\usepackage{subcaption}

\newtheorem{theorem}{Theorem}[section]

\newtheorem{proposition}[theorem]{Proposition}
\newtheorem{lemma}[theorem]{Lemma}
\newtheorem{corollary}[theorem]{Corollary}

\theoremstyle{definition}

\newtheorem{example}[theorem]{Example}
\newtheorem{remark}[theorem]{Remark}

\begin{document}

\newcommand{\beq}{\begin{equation}}  
\newcommand{\eeq}{\end{equation}}  
\newcommand{\bea}{\begin{eqnarray}}  
\newcommand{\eea}{\end{eqnarray}}  
\newcommand\la{{\lambda}}   
\newcommand\La{{\Lambda}}   
\newcommand\ka{{\kappa}}   
\newcommand\al{{\alpha}}   
\newcommand\be{{\beta}} 
\newcommand\ze{{\eta}} 
\newcommand\zet{{\nu}} 
\newcommand\gam{{\gamma}}     
\newcommand\om{{\omega}}  
\newcommand\tal{{\tilde{\alpha}}}  
\newcommand\tbe{{\tilde{\beta}}}   
\newcommand\tla{{\tilde{\lambda}}}  
\newcommand\tmu{{\tilde{\mu}}}  
\newcommand\si{{\sigma}}  
\newcommand\lax{{\bf L}}    
\newcommand\mma{{\bf M}}    
\newcommand\rd{{\mathrm{d}}}  
\newcommand\tJ{{\tilde{J}}}  
\newcommand\ri{{\mathrm{i}}} 

\newcommand{\N}{{\mathbb N}}
\newcommand{\Q}{{\mathbb Q}}
\newcommand{\Z}{{\mathbb Z}}
\newcommand{\C}{{\mathbb C}}
\newcommand{\R}{{\mathbb R}}

\newcommand\tI{{\tilde{\mathcal{I}}}}

\newcommand\SH{{\mathcal{S}}}

\newcommand\tr{{{\mathrm{tr}}\,}}


\title{Some integrable maps and their Hirota bilinear forms} 
\author{A.N.W. Hone and  T.E. Kouloukas 
\\School of Mathematics, Statistics \& Actuarial Science \\ University of Kent, UK \\  \\ 
G.R.W. Quispel 
\\ 
Department of Mathematics and Statistics  \\
La Trobe University, Bundoora VIC 3086, Australia
}
\maketitle

\begin{abstract}
We introduce a two-parameter family of birational maps, which reduces to a family previously found by Demskoi, Tran, van der Kamp and Quispel (DTKQ) when one of the parameters is set to zero. 
The study of the singularity confinement pattern for these maps  leads to the introduction of a tau function satisfying  a homogeneous 
recurrence which  has the Laurent property, and the tropical (or ultradiscrete) analogue of this 
homogeneous recurrence confirms the quadratic degree growth found empirically by Demskoi et al. 
We prove that the tau function also  satisfies two different 
bilinear equations, each of which is a reduction of the Hirota-Miwa equation (also known as the discrete KP equation, or the octahedron recurrence). Furthermore, these bilinear equations are related to reductions of particular two-dimensional integrable lattice equations, of discrete KdV or discrete Toda type. These connections, as well as the cluster algebra structure of the bilinear equations, allow a direct construction of Poisson brackets, Lax pairs and first integrals for the birational maps. As a consequence of the latter results, we show how each member of the family  can be lifted to a system 
that is  integrable in the  Liouville sense, clarifying observations made previously in the original DTKQ case. 
\end{abstract}

\section{Introduction} 

In recent work \cite{DTKQ},  Demskoi, Tran, van der Kamp and Quispel (DTKQ) introduced a one-parameter family of birational maps, 
given by the $N$th-order difference equation  
\beq\label{dtkq} 
\Big( u_n+u_{n+1}+\ldots +u_{n+N} \Big) \, u_{n+1}u_{n+2}\cdots u_{n+N-1}=\al, 
\eeq 
for each integer $N\geq 2$. It was shown that the equation (\ref{dtkq}) admits $\left \lfloor{\frac{N+1}{2}} \right \rfloor$ 
independent first integrals, 
explicitly derived in terms of 
multi-sums of products, 
and from a conjectured formula for the degrees $d_n$ of the iterates (quadratic in the index $n$)  it was inferred that 
$\lim_{n\to\infty} n^{-1}\log d_n =0$ for each $N$, 
so that the corresponding map should have vanishing algebraic entropy in the sense of \cite{hv}.  These results suggested that (\ref{dtkq}) should correspond to a
finite-dimensional  system that is  integrable in the  Liouville sense \cite{maeda,veselov}. 

For all $N$ it was noted in \cite{DTKQ} that, up to orientation,   the map 
$$\sigma: \quad (u_0,\ldots, u_{N-1})\mapsto (u_1,\ldots,u_N)$$ defined 
by (\ref{dtkq}) preserves the canonical volume form
$$ 
\Omega =   \rd u_0\wedge \cdots \wedge \rd u_{N-1}, 
$$ 
so that $\sigma^*\Omega =(-1)^N \Omega$ for each $N$. This means that for $N=2$ the map $\sigma$ is symplectic, while for $N=3$ it can be reduced to an anti-symplectic (orientation-reversing) map by restricting to a level set 
of one of the first integrals; so this is enough to imply Liouville integrability for $N=2,3$. 
However, in the absence of a suitable symplectic or Poisson structure, it is not possible to assert that the maps 
are Liouville integrable for $N\geq 4$. 


In this paper we consider a generalization of (\ref{dtkq}) with  two parameters, given by  
\begin{equation} \label{equ}
\left(\sum_{j=0}^{N}u_{n+j} +\beta\right)\prod_{k=1}^{N-1} u_{n+k} = \alpha,  
\end{equation}
and show that this slight extension allows a natural interpretation of the observations made in \cite{DTKQ}, in terms of 
a Liouville integrable system with $\left\lfloor{\frac{N}{2}} \right\rfloor$ degrees of freedom. In fact, 
for even $N$ we shall show that the solutions of (\ref{equ}) are related to the $(N,-1)$ periodic reduction of Hirota's discrete KdV equation, for which  Liouville integrability was proved in \cite{hkqt}, 
while for odd $N$ they are connected to reductions of a discrete Toda lattice, considered recently in \cite{hkw}. Furthermore, for all $N$ these maps are linked to reductions of the Hirota-Miwa equation (also known as the discrete KP equation, or the octahedron recurrence), which connects them to certain cluster algebras \cite{FH} and leads 
to some associated symplectic maps, referred to as U-systems \cite{hi}.  

The original derivation of the equation (\ref{dtkq}) was based on the fact that it is dual to a linear difference equation of 
order $N$, in the sense introduced in \cite{qcr}:  the pair of dual equations have  a first integral in common, and each equation appears as an integrating factor for 
the other one in the total difference of this first integral. The same observation applies to the more general version 
(\ref{equ}),   by introducing 
\beq\label{zeta} 
\zeta = \left( \sum_{j=0}^{N-1} u_{n+j} +\be \right) \Big( \al - \prod_{k=0}^{N-1}u_{n+k}\Big). 
\eeq 
The latter  quantity is seen to be a first integral of (\ref{equ})  by noting the identity 
$$ 
\Delta \zeta = (u_{n+N}-u_n) \left( \al -  \Big(\sum_{j=0}^{N}u_{n+j} +\beta\Big)\prod_{k=1}^{N-1} u_{n+k}\right),  
$$ 
in terms of the total difference 
$\Delta={\cal S} -1$, with 
$\cal S$ being the shift operator such that ${\cal S} F_n=F_{n+1}$ (where $F_n$ is any function of $n$). 
The linear factor above shows that (\ref{equ}) is dual to the linear equation $u_{n+N}-u_n =0$, in the same 
sense that (\ref{dtkq}) is, and for $\be=0$ the quantity $\zeta$ reduces to the first integral that is denoted by
the same 
letter in   \cite{DTKQ}. 
  
The key to our results is the use of the singularity confinement pattern of (\ref{equ}) to 
obtain its ``Laurentification'' \cite{hk}, i.e. a lift to a higher-dimensional system which has the Laurent property in the sense of \cite{fz}. The solution of the higher-dimensional system can then be shown  to satisfy a Hirota bilinear equation (in fact, two bilinear equations for each $N$). Our main result, from which all the rest can be derived, is the following.

\begin{theorem}\label{maintau} 
 Suppose that 
\beq\label{tausub} 
u_n = \frac{\tau_{n+3}\tau_n}{\tau_{n+2}\tau_{n+1}} 
\eeq 
is a solution of (\ref{equ}). Then $\tau_n$ satisfies the bilinear equation 
\beq\label{hir1}
\tau_{n+N+2}\tau_n = \gam_n \, \tau_{n+N+1}\tau_{n+1} +\al \, \tau_{n+N}\tau_{n+2},  
\eeq  
where the quantity $\gam_n$ is  
2-periodic, 
that is 
$$ 
\gam_{n+2}=\gam_n \qquad \forall n;  
$$ 
and conversely, the equation (\ref{hir1})  for  $\tau_n$, with the 2-periodic coefficient $\gam_n$,  has a first 
integral $\beta$ such that $u_n$ given by (\ref{tausub}) satisfies (\ref{equ}). Moreover, if $u_n$ is given 
by (\ref{tausub}),   then when $N$ is $\mathrm{even}$  (\ref{equ}) has  a first integral $K$ such that $\tau_n$ satisfies 
\beq\label{hir2}
\tau_{n+2N+1}\tau_n =-\al \, \tau_{n+2N}\tau_{n+1}+K \,   \tau_{n+N+1}\tau_{n+N},  
\eeq  
while for $N$ $\mathrm{odd}$  (\ref{equ}) has a first integral $\bar K$ such that 
\beq\label{hir3}
\tau_{n+2N+2}\tau_n = \al^2 \, \tau_{n+2N}\tau_{n+2} + \bar{K} \, \tau_{n+N+1}^2.   
\eeq  
\end{theorem}  

An outline of the paper is as follows. 
In the next section, we explain how we originally obtained  the above result, using the singularity confinement method (or 
an arithmetical analogue of it) to find a Laurentification of (\ref{equ}), given by a multilinear equation 
for a tau function (equation (\ref{multi}) in section 2). 
We also present a tropical (ultradiscrete) version of the multilinear equation, as well as 
a corresponding tropical version of (\ref{equ}), 
and show how this can be used to obtain an explicit formula for degree growth 
 (quadratic in $n$).  
For any fixed $N$, it is then possible to derive the bilinear equation (\ref{hir1}), as well as  (\ref{hir2}) or  (\ref{hir3}), 
either numerically (with specific initial data) or symbolically (working with rational functions of initial data).  
The complete proof of Theorem \ref{maintau}, for arbitrary $N$,  is reserved until section 3, where we begin by deriving  (\ref{hir1})  
via a modified version of (\ref{equ}) (see equation (\ref{eqw}) below), before 
treating the rest of the result and its detailed consequences for even/odd $N$ separately. The 
interpretation in terms of Liouville integrability is naturally achieved by considering a lift of (\ref{equ}) to dimension 
$N+1$, obtained by eliminating the parameter $\be$: this yields equation (\ref{ushift}) in section 3. 
A schematic picture of the connections between the main equations involved, valid for arbitrary $N$, is 
provided by the following diagram: 
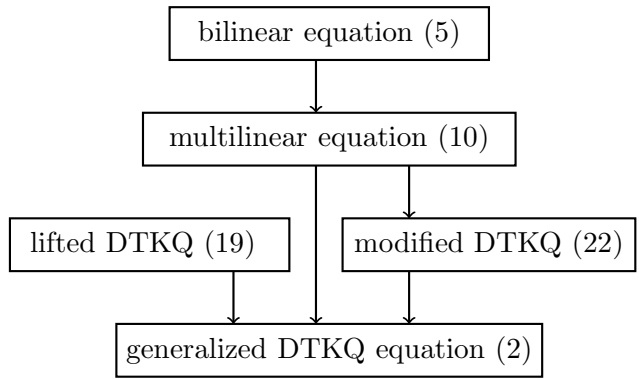
\begin{center}
\begin{tikzpicture}[scale=.7]

\draw[thick]  (5,4) rectangle (11,5)
(4.5,2) rectangle (11.5,3)
(8.25,0) rectangle (13.75,1)
(2,0) rectangle (7.25,1)
(4,-2) rectangle (12,-1);
                      
 \draw   (8,4.5) node {bilinear equation (\ref{hir1})}
 (8,2.5) node {multilinear equation (\ref{multi})}
 (11,0.5) node {modified DTKQ  (\ref{eqw})}
 (4.5,0.5) node {lifted DTKQ  (\ref{ushift})}
  (8,-1.5) node {generalized DTKQ equation  (\ref{equ})};

\draw[thick,->]  (7.75,4) -- (7.75,3);	 
\draw[thick,->]  (7.75,2) -- (7.75,-1);  
\draw[thick,->]  (9.5,2) -- (9.5,1); 
\draw[thick,->]  (6.2,0) -- (6.2,-1);
\draw[thick,->]  (9.5,0) -- (9.5,-1); 
\end{tikzpicture}
\end{center}
%
The vertical arrows above denote maps between solutions of an equation and the one below it. 
The other results in section 3 are based on the connection between bilinear equations and cluster algebras, as 
explained in \cite{FH}, which leads to a Poisson structure for the lifted DTKQ equation (\ref{ushift}). For $N$ even, both 
 bilinear equations (\ref{hir1}) and (\ref{hir2}) reveal the connection with reductions of Hirota's discrete KdV equation;
while  for $N$ odd,  the second bilinear equation (\ref{hir3}) leads to a link with reductions of a 
discrete time Toda equation, as 
well as an associated B\"acklund transformation (or BT, in the sense of \cite{kuskly}). 
In order to illustrate these  results, we provide full details for the particular even case $N=4$  in section \ref{even}, and 
for the odd case $N=5$ in section \ref{odd}, before finishing with some conclusions.    

\section{Singularity confinement 
and Laurentification} 

In this section we 
describe the experimental approach which led us to 
Theorem \ref{maintau}. 

The first relevant tool here is the singularity confinement  test, which was introduced in \cite{grp} as a heuristic  method for
identifying discrete systems that may be integrable. In its original form, this method has the drawback that many systems 
with confined singularities have positive algebraic entropy \cite{hv}, but recently singularity confinement has been refined to include information  about deautonomization, which renders it a more effective tool \cite{con1,con2,con3}. If we apply the basic singularity confinement test to (\ref{equ}) for a few small values of $N$, then in all cases we find that the 
singularity pattern is 
\beq\label{sing} 
\ldots, \epsilon, \epsilon^{-1}, \epsilon^{-1},\epsilon, \ldots,  
\eeq  
where the latter is the leading power of $\epsilon$ when the singularity is approached as $\epsilon\to 0$. 

In fact, in order to see the singularity pattern, we do not really need to apply the singularity confinement test per se, but 
rather an arithmetical version of it, by considering orbits of (\ref{equ}) defined over $\Q$, for rational values of the 
parameters $\al,\be$. This can be turned into a semi-numerical method for measuring the growth of complexity \cite{aba}, 
with the rate of growth of the logarithmic heights of the iterates being taken as a measure of entropy \cite{halburd}. 
Furthermore, if the map is defined over $\Q$, then one can consider reduction  modulo a prime $p$, in which case the appearance of a singularity 
at some iterate $u_n\in \Q$ means that the $p$-adic norm $\left|u_n\right|_p>1$, and the $p$-adic expansion of the iterates 
(expanding in powers of $p$) is analogous to the expansion in powers of $\epsilon$ in the usual 
singularity confinement test (see \cite{kanki1,kanki2} for an application of this idea). 

To see this method in practice, consider 
(\ref{equ}) for $N=5$ with $\al=3$, $\be=-10$, which gives the recurrence 
$$ 
u_{n+5} = 10 -(u_n+u_{n+1}+u_{n+2}+u_{n+3}+u_{n+4}) + \frac{3}{u_{n+1}u_{n+2}u_{n+3}u_{n+4}}, 
$$ 
defined over $\Q$, 
and choose the five rational  initial data $u_0=u_1=u_2=u_3=1$, $u_4=4$. 
The sequence continues as  
\beq\label{ratseq} 
\frac{11}{4}, \frac{23}{44}  , \frac{316}{253}, 
\frac{1628} {1817}, \frac{ 7153}{2923}, \frac{194735}{46028}, \frac{2800493}{3066460}, 
\frac{115286767}{186573385}, \ldots ,
\eeq 
and if we factorize each of the above terms then we find 
$$ 
\tfrac{11}{2^2}, \tfrac{23}{2^2\cdot 11}, \tfrac{ 2^2\cdot 79}{11\cdot 23}, 
\tfrac{2^2\cdot 11\cdot 37}{23\cdot 79} , \tfrac{ 23\cdot 311}{37\cdot 79}, 
\tfrac{5\cdot 17\cdot 29 \cdot 79} {2^2\cdot 37\cdot 311}   
, \tfrac{37 \cdot 75689}{2^2\cdot  5\cdot 17\cdot 29\cdot 311}, 
\tfrac{59\cdot 61\cdot 103\cdot 311}{5\cdot 17\cdot 29\cdot 75689}, \ldots ,
$$
which can be taken $\bmod \, p$ for $p=11,23,37,79,311$ etc. to reveal the singularity pattern 
$p^{-1},p,p,p^{-1}$ at  leading order, in accordance with (\ref{sing}) (while the choice $p=2$ is special here because the coefficient $\be$ vanishes $\bmod \,2$ in this example).    

There is a close link between singularity confinement for discrete systems and the Laurent property \cite{honeconf,mase1}. 
In the context of integrability, the Laurent property appears at the level of Hirota bilinear equations: 
the Hirota-Miwa (discrete KP) equation can be derived from mutations in a cluster algebra   \cite{okubo}, 
which means  that it has the Laurent property, and this property is inherited by its reductions to recurrences of Somos (or Gale-Robinson) type \cite{fz,mase2}.   
Furthermore, it seems likely that any birational map with confined singularities can be lifted to a 
higher-dimensional ``Laurentified'' system, i.e. one that has the Laurent property. 
In specific examples, Laurentification in this sense has been obtained by passing to homogeneous coordinates \cite{viallet},  or by using 
recursive factorization \cite{hk,hhkq}. 

For the equations (\ref{equ}), 
regardless of the means by which we obtain 
the singularity pattern, the form of (\ref{sing})  immediately suggests that we should try the substitution (\ref{tausub})  
in order to obtain the Laurentification. 
\begin{proposition}\label{lify} 
Given the substitution (\ref{tausub}),  $u_n$ is a solution of (\ref{equ}) 
whenever 
 the tau function $\tau_n$ satisfies 
the multilinear relation 
\beq\label{multi} 
\begin{array}{rl} 
\tau_{n+N+3} \tau_{n+N}^2\prod_{j=1}^{N-1}\tau_{n+j} & =   
\al \tau_{n+3}\tau_{n+N}\prod_{j=2}^{N+1} \tau_{n+j}  
-\be 
 \prod_{j=1}^{N+2}\tau_{n+j}\\ &
-\tau_n\tau_{n+3}\prod_{j=3}^{N+2} \tau_{n+j}   
-\sum_{k=1}^{N-1}\tau_{n+k}\tau_{n+k+3} \underset{j\neq k+1,k+2}{\prod_{j=1}^{N+2}}\tau_{n+j}, 
\end{array} 
\eeq
which is of order $N+3$ and homogeneous of degree $N+2$.  
For each $N\geq 2$ the recurrence (\ref{multi}) has the Laurent property, i.e. the iterates are 
Laurent polynomials in the initial data $\boldsymbol{ \tau}=(\tau_0,\tau_1,\ldots, \tau_{N+2})$. More precisely,  
$$ 
\tau_n \in \Z[\al,\be, \tau_0^{\pm 1},\ldots, \tau_{N+2}^{\pm 1}] \qquad \forall n\in\Z.   
$$
\end{proposition} 
\begin{proof} To prove the Laurent property for the equation (\ref{multi}), we make use of the bilinear relation (\ref{hir1}) in Theorem \ref{maintau}, of which an 
independent proof is given in the next section. Let ${\cal R}:=\Z[\al,\be, \tau_0^{\pm 1},\ldots, \tau_{N+2}^{\pm 1}]$. The period 2 coefficient appearing in 
(\ref{hir1}) takes two distinct values, given by 
$$ 
\gam_0=({\tau_{N+1}\tau_1})^{-1}\Big({\tau_{N+2}\tau_0-\al \tau_N\tau_2}\Big), \quad 
\gam_1=({\tau_{N+2}\tau_2})^{-1}\Big({\tau_{N+3}\tau_1-\al \tau_{N+1}\tau_3}\Big) \in {\cal R}, 
$$ 
using the fact that $\tau_{N+3}\in{\cal R}$, which follows directly from (\ref{multi}). So the iterates of (\ref{multi}) coincide with those of (\ref{hir1}), subject to 
fixing $\gam_0,\gam_1$ as above. Now we can make use of Proposition 5.4 in \cite{mase2}, 
which implies that the nonautonomous Somos recurrence (\ref{hir1}) has the Laurent property, meaning that 
$$ 
\tau_{n}\in  
\Z[\al,\gam_0, \gam_1, \tau_0^{\pm 1},\ldots, \tau_{N+1}^{\pm 1}] 
$$
for all $n$. Upon substituting $\gam_0,\gam_1\in{\cal R}$ into the Laurent polynomials obtained from (\ref{hir1}), the result follows. 
\end{proof}

For $\be=0$, corresponding to (\ref{dtkq}), particular cases of the preceding result have been proved before. The case $N=2$ of the recurrence (\ref{multi})  with $\be=0$ was found 
previously by the recursive factorization method: this is Theorem 8 in \cite{hk}, while Theorem 10 in the same paper corresponds to the case $N=3$, but with the inclusion of certain periodic coefficients; and some results for general $N$ are found  in \cite{hamad}. 
For any particular $N$ it can be verified directly with computer algebra that the Laurent property holds, by applying a method attributed to Hickerson,  
that is described in \cite{rob};      
but a direct proof along these lines for all $N$ is not so straightforward.   

\begin{example}\label{mul5} 
When $N=5$, the recurrence (\ref{multi}) becomes 
\beq\label{multi5} 
\begin{array} {rcl} 
\tau_{8}\tau_{5}^2\tau_{4}\tau_{3}\tau_{2} \tau_{1} & = &   
\al \tau_{6}\tau_{5}^2\tau_{4}\tau_{3}^2\tau_{2}  
-\be 
\tau_{7}  \tau_{6} \tau_{5}  \tau_{4} \tau_{3} \tau_{2} \tau_{1}  
 -\tau_{7}  \tau_{6} \tau_{5}  \tau_{4} \tau_{3}^2  \tau_{0} 
 - \tau_{7}  \tau_{6} \tau_{5}  \tau_{4}^2 \tau_{1}^2  \\
&& 
 - \tau_{7}  \tau_{6} \tau_{5}^2  \tau_{2}^2 \tau_{1} 
 - \tau_{7}  \tau_{6}^2 \tau_{3}^2  \tau_{2} \tau_{1} 
 - \tau_{7}^2  \tau_{4}^2 \tau_{3} \tau_{2}  \tau_{1} ,
\end{array}  \eeq 
where we have set the index $n\to 0$ for brevity.
\end{example}

The Laurent property means that the iterates of (\ref{multi}) can be written as 
\beq\label{lpoly} 
\tau_n =\frac{   \mathrm{N}_n ( \boldsymbol{\tau})  }
{ \boldsymbol{\tau}^{{\bf d}_n} } , 
\eeq 
where the numerator $\mathrm{N}_n$ is a polynomial in the initial data that is not divisible by any of the variables 
$\tau_0,\tau_1,\ldots,\tau_{N+2}$, while $ \boldsymbol{\tau}^{{\bf d}_n}$ denotes the Laurent monomial in these variables specified by 
the denominator vector ${\bf d}_n$, an integer vector whose components give the exponents for each variable.  
Due to the homogeneity of (\ref{multi}), the degree growth of the iterates can be determined from that of the
denominators. If we further assume  that there are no cancellations between numerators and denominators on the 
right-hand side, then (as is well known in the context of cluster algebras \cite{f1,f2}), the denominator vector 
${\bf d}_n$ satisfies the max-plus tropical (or ultradiscrete) analogue of (\ref{multi}), which takes the form 
\beq\label{tropm}   
{\bf d}_{n+N+3}+  2{\bf d}_{n+N}+\sum_{j=1}^{ N-1}{\bf d}_{n+j}  = 
\max\left(  {\bf d}_{n+3}+ {\bf d}_{n+N}+\sum_{j=2}^{N+1} {\bf d}_{n+j},   
 \sum_{j=1}^{N+2}{\bf d}_{n+j}, \ldots \right) , 
\eeq 
where each of the omitted terms in the max 
corresponds to one of the terms on the right-hand side
of (\ref{multi}). The vector form of (\ref{tropm}) means that each component of ${\bf d}_n$ satisfies the same tropical equation.  

\begin{example}\label{trop5} 
When $N=5$, the tropical version of (\ref{multi5}) can be written as 
$$
 {\bf d}_{8}+2{\bf d}_{5}+{\bf d}_{4}+{\bf d}_{3}+{\bf d}_{2}+{\bf d}_{1} 
=\max \Big({\bf d}_6+2{\bf d}_5+{\bf d}_4+2{\bf d}_3+{\bf d}_2,\sum_{j=1}^7 {\bf d}_j, 
\sum_{j=4}^7 {\bf d}_j
+2{\bf d}_3+{\bf d}_0, \hat{M}  \Big),  
$$
where 
$$ 
\hat{M}= {\bf d}_7+{\bf d}_1+ \max\Big({\bf d}_6+{\bf d}_5+2{\bf d}_4+{\bf d}_1, 
{\bf d}_6+2{\bf d}_5+2{\bf d}_2, 
2{\bf d}_6+2{\bf d}_3+{\bf d}_2, 
{\bf d}_7+2{\bf d}_4+{\bf d}_3+{\bf d}_2\Big); 
$$ 
once again we have set $n\to 0$ for brevity. 
\end{example}  

From the explicit form of the above equation, the task of solving (\ref{tropm}) for general $N$ looks like it might be a formidable one, but in fact there is an enormous simplification that can be made. The point is that the substitution 
(\ref{tausub}) 
that lifts (\ref{equ}) to (\ref{multi}) also has a tropical analogue, which allows (\ref{tropm}) to be reduced to a max-plus 
version of (\ref{equ}), and the latter turns out to have a very simple behaviour: all solutions reach a fixed point after 
finitely many steps.    

\begin{proposition} \label{troplim} 
Given the substitution 
\beq\label{tropsub}
U_n=d_{n+3}-d_{n+2}-d_{n+1}+d_n, 
\eeq
the quantity 
$U_n$ is a solution of a tropical 
version  of 
(\ref{equ}),  
given by 
\beq 
 \label{tropequ}
U_{n+N}=\Big[\max(-S_n,U_n,U_{n+1},\ldots,U_{n+N-1})\Big]_+, \qquad S_n=\sum_{j=1}^{N-1}U_{n+j}
\eeq 
(with $[x]_+$ denoting $\max (x,0)$ for $x\in\R$), 
whenever $d_n$ is a scalar solution of (\ref{tropm}). 
Moreover, for any choice of real initial data for (\ref{tropequ}) there exists $C\geq 0$ 
and an integer $m\geq 0$ such that 
\beq\label{fixp} 
U_n=C \qquad \forall n\geq m. 
\eeq 
\end{proposition} 
\begin{proof} 
The max-plus equation (\ref{tropequ}) is obtained from (\ref{equ}) by solving for the highest iterate $u_{n+N}$ 
and replacing $(+,\times)$ with $(\max , +)$ in the usual way (setting all coefficients to 1), 
while  a direct 
calculation verifies that if $d_n$ satisfies the scalar version of (\ref{tropm}) then $U_n$ given by (\ref{tropsub}) is 
a solution of (\ref{tropequ}). Now given an $N$-tuple of real initial values $(U_0,\ldots , U_{N-1})$, we consider the iteration of the 
equivalent map 
in $\R^N$ given by 
\beq\label{tropmap}
(U_0,\ldots , U_{N-1})\mapsto (U_0',\ldots , U_{N-1}')
\eeq 
where 
$$ U_j'=U_{j+1}\,\, \mathrm{ for}\,\, 0\leq j\leq N-2, \qquad  
  U_{N-1}'=[\max(-S,U_0,\ldots,U_{N-1})]_+, \quad S=\sum_{j=1}^{N-1}U_j.
$$
The initial data can be divided into four disjoint subsets of $\R^N$, defined by 
$$ 
\begin{array} {lcl} 
\mathrm{(i)} & 
S\geq 0, &
\max_{j\in \{ 1,\ldots,N-1\}}(U_j) 
\geq U_0 ;  
\\ 
\mathrm{(ii)} & 
S\geq 0, &
U_j < U_0 
\quad
\forall j\in \{ 1,\ldots,N-1\} ; 
\\ 
\mathrm{(iii)} & 
S <0, &
U_j < -S  
\,\,\,\,
\forall j\in \{ 0,\ldots,N-1\} ; 
\\ 
\mathrm{(iv)} & 
S< 0, &
\max_{j\in \{ 0,\ldots,N-1\}}(U_j) \geq -S.
\end{array} 
$$ 
By examining each of these regions, it follows that the quantity $S$ is non-decreasing under the action of the map (\ref{tropmap}), and after finitely many iterations 
it 
attains a maximum value $S=NC$ at a fixed point $U_j=C\geq 0$ for $0\leq j\leq N-1$. 
To see this, begin by considering initial data lying in region (i). In that case, the maximum of the  $U_j$ 
is attained at some $k\in\{1\,\ldots, N-1\}$, and $U_{N-1}'=U_k=C\geq 0$, with 
$S'=\sum_{j=1}^{N-1}U_j'=S+U_{N-1}'-U_1\implies S'-S=U_k-U_1\geq 0$. All subsequent iterations stay in this region, $C$  remains the maximum 
value, and all components take this same  value 
after a finite number of steps. Next, take initial data in region (ii), to find $U_{N-1}'=U_0$ and $S'-S=U_0-U_1>0$. In that 
case,  the maximum value is $C=U_{N-1}'>U_0'=U_1$, and so region (i) is reached after a single step. For case (iii), 
one step of (\ref{tropmap})  gives $U_{N-1}'=-S>0$, so $S'-S=-S-U_1>0$, and hence 
$\max_{j\in\{0,\ldots,N-1\}}(U_j')= U_{N-1}'>-S'$ which means that either region (iv) is attained 
when $S'<0$, or otherwise $S'\geq 0$ and region (i) has been reached instead. Finally, starting off in region (iv) 
gives $U_{N-1}'=U_k=C\geq -S>0$ for some $k$, and $S'-S=U_k-U_1\geq 0$. If the sum $S'<0$, a finite number of subsequent steps remain in region (iv), with $C$ as the maximum value, until this sum changes sign, so that eventually region (i) is reached, and the proof is complete. 
\end{proof}

Upon comparing the substitution (\ref{tropsub}) with (\ref{fixp}), the explicit form of the scalar solution of (\ref{tropm}) is obtained immediately, for sufficiently large $n$. 

\begin{corollary}\label{quaddeg}
For any real $(N+3)-tuple$ of initial values $(d_0,d_1,\ldots,d_{N+2})$, 
there exist real parameters $A,\bar{A},B,C$ with $C\geq 0$ and an integer $m\geq 0$ 
such that the   
solution of 
the scalar version of (\ref{tropm}) is given by  
\beq\label{quadg} 
d_n =\frac{1}{4}Cn^2 +Bn+A + \bar{A}(-1)^n \qquad \forall n\geq m.  
\eeq\end
{corollary} 

\begin{remark}
The iterates of (\ref{multi}) are given by (\ref{lpoly}), where the 
vector ${\bf d}_n$ has components $d_n^{(j)}$, giving the degree of the exponent of $\tau_j$, for $0\leq j\leq N+2$. 
For each $j$ in this range, this gives the initial data $d_k^{(j)}=-\delta_{jk}$, $0\leq k\leq N+2$, for a scalar solution of 
(\ref{tropm}). 
Hence the degrees of the denominators  of the iterates of (\ref{multi}) can be calculated exactly, 
while (by homogeneity) the degree of each numerator is one more than the degree of the denominator. 
In particular, $C=0$ in the solution  for $j=0$, and in fact $d_n^{(0)}=0$ $\forall n\geq1$, 
since there is no division by $\tau_0$ when  (\ref{multi}) is iterated forwards, and $C=1$ in the solution for 
$d_n^{(j)}$, $1\leq j\leq N-1$. With more detailed analysis, the precise degree growth of (\ref{equ}) can also be derived, but it follows from (\ref{quadg}) that it must be quadratic, which  is consistent with the empirical results found in \cite{DTKQ} for the case $\beta=0$. 
More detailed results for symmetric QRT maps, which include the case $N=2$ of (\ref{equ}), are given in \cite{hhkq}.  
\end{remark}

The Laurent property, combined with the quadratic degree growth of the iterates of (\ref{multi}), suggests that there should also be bilinear relations satisfied by the terms. Thus we can apply the method described in \cite{hones6}, starting with 
particular numerical values of the initial conditions and looking for the smallest integer $q$ such that a matrix $\mathrm{M}$ 
of size $\left\lfloor{\frac{q}{2}}\right\rfloor+1$, with entries $\mathrm{M}_{ij}=\tau_{q+i-j}\tau_{i+j-2}$, has vanishing determinant: this corresponds to a bilinear relation of minimal order, with constant coefficients. By considering matrix entries of the form $\mathrm{M}_{ij}=\tau_{q+\ell i-j+k}\tau_{\ell i+j+k-2}$, for different choices of offset $k$, one can also obtain minimal order relations whose coefficients have period $\ell>1$. 

To illustrate the method, we pick $N=5$ with $\alpha=3$, $\beta=-10$, and use (\ref{multi5})  to generate the  
particular sequence that starts from $\tau_0=\tau_1=\ldots = \tau_6=1$, $\tau_7=4$,  beginning  
$$ 
1,1,1,1,1,1,1,4,11,23, 79, 148,1244,9860,75689,370697,\ldots, 
$$
whose ratios $\tau_{n+3}\tau_n/(\tau_{n+2}\tau_{n+1})$ produce (\ref{ratseq}). By considering 
bilinear relations with constant coefficients, the minimal relation is found to be of order $q=12$, 
corresponding to the matrix $\mathrm{M}$ with entries   $\mathrm{M}_{ij}=\tau_{12+i-j}\tau_{i+j-2}$. 
However, all but three of the entries in a vector spanning the one-dimensional kernel of   $\mathrm{M}$  are zero, and it is 
sufficient to  take 
the $3\times 3$ minor 
$$ 
\mathrm{M}' = \left(\begin{array}{ccc} \tau_{12}\tau_0 & \tau_{10}\tau_2 & \tau_{6}^2 \\ 
\tau_{13}\tau_1 & \tau_{11}\tau_3 & \tau_{7}^2 \\ 
\tau_{14}\tau_2 & \tau_{12}\tau_4 & \tau_{8}^2
\end{array}\right) = 
\left(\begin{array}{ccc} 1244 & 79 & 1 \\ 
9860 & 148 & 16 \\ 
75689 & 1244 & 121 
\end{array}\right), \qquad \det \mathrm{M}' =0, 
$$ 
whose kernel is spanned by the vector $(1,-9,-533)^T$, corresponding to the bilinear relation (\ref{hir3}) with $N=5$, 
$\al = 3$ and the particular value $\bar{K}=533$ for the first integral. 
With this same numerical sequence, one can also obtain relations with periodic coefficients, starting from period $\ell=2$, by taking the matrix with entries of the form  $\mathrm{M}_{ij}=\tau_{q+2 i-j+k}\tau_{2 i+j+k-2}$, so that the 
minimum order relation has $q=7$, and for $k=-1,0$ one has $3\times3$ minors with  
$$ 
\det\left(\begin{array}{ccc} 4 & 1 & 1 \\ 
23 & 11 & 4 \\ 
148 & 79 & 23 
\end{array}\right) =0 
= 
\det\left(\begin{array}{ccc} 
11 & 4 & 7 \\ 
79 & 23  & 11 \\
1244 & 148 & 316 
\end{array}\right) . 
$$ 
Each of the two matrices above has a one-dimensional kernel, spanned by $(1,-1,-3)^T$, $(1,-2,-3)^T$ respectively, 
corresponding to the three-term bilinear relation (\ref{hir1}) with   $N=5$, 
$\al = 3$, $\gam_0=1$, $\gam_1=2$. 
  
For fixed $N$, once bilinear relations have been obtained for 
one or more particular numerical sequences of values of $\tau_n$, 
these relations can then be checked for arbitrary initial data and coefficients by symbolic computations with a 
computer algebra package. Such computations provide a computer-assisted proof of Theorem \ref{maintau}, for any specific choice of $N$, but to prove it for all $N$ requires some general arguments, presented in the next section. 

Sequences generated by certain bilinear recurrences of Somos type also admit further bilinear relations of higher order (see 
\cite{svdp}, for example), with the coefficients being first integrals, and this has been used to obtain first integrals for 
four-term
Somos-6 and Somos-7 recurrences, in \cite{hones6} and \cite{FH}, respectively.  Another approach to finding first integrals, based on reduction of conservation laws for the discrete KP or BKP equations, was used in \cite{mq}. However, in what follows we will apply 
the method  in \cite{hkw}, 
obtaining first integrals from Lax pairs arising by reduction of lattice equations 
(discrete KdV, discrete Toda, and/or discrete KP). 
 
\section{Proof and consequences of the main theorem} 


In order to understand how the solutions of   (\ref{equ}) are related to certain Liouville integrable systems, it is helpful 
to consider the equation of order $N+1$ obtained by eliminating $\be$: 
\beq\label{ushift} 
E_n[u]:=u_{n+N+1}-u_n +\al\left( \frac{1}{\prod_{j=1}^{N-1}u_{n+j}}- \frac{1}{\prod_{j=2}^{N}u_{n+j}}\right)=0. 
\eeq 
To see how this arises, one can solve (\ref{equ}) for $\be$, which gives 
$$ 
\be = \frac{\al}{\prod_{j=1}^{N-1}u_{n+j}}-\sum_{j=0}^N u_{n+j}, 
$$ 
and then apply  
the total difference operator  $\Delta$  to both sides;
from  this it follows that $\be$ defined as above is a first integral for (\ref{ushift}). For all $N$, it can also be checked that 
(\ref{ushift}) has the first integral 
\beq\label{zehat}
\hat{\zeta} = \prod_{j=0}^{N}u_{n+j}+\al\sum_{k=1}^{N-1}u_{n+k}, 
\eeq 
which is related to the first integral (\ref{zeta}) of  (\ref{equ}) by 
$\hat{\zeta}=\zeta -\al\be$. 
 
For the  proof of the first part of Theorem \ref{maintau}, 
it is 
convenient to introduce another lift of (\ref{equ}) to dimension $N+1$, defined via  (\ref{tausub}) by 
setting 
\beq\label{wsub} 
\pi_1: \quad 
u_n=w_nw_{n+1} \qquad \mathrm{where} \quad w_n=\frac{\tau_n\tau_{n+2}}{\tau_{n+1}^2}
.
\eeq
This can be regarded as an intermediate step between (\ref{equ}) and (\ref{multi}). 
Upon making the    substitution (\ref{wsub}), we find that 
(\ref{equ}) produces the equation 
\begin{equation} \label{eqw}
{\cal E}_n [w]:=\sum_{j=0}^{N}w_{n+j}w_{n+j+1}  +\beta -\frac{\al}{\prod_{j=1}^{N-1} w_{n+j}w_{n+j+1}} =0. 
\end{equation}
If the substitution (\ref{wsub}) is interpreted as a Miura map, then (\ref{eqw}) can be regarded as a modified version 
of the generalized DTKQ equation (\ref{equ})

\begin{lemma} The quantity given in terms of $w_n$ and shifts by 
\begin{equation} \label{gamma}
\gamma_n [w]  = 
\prod_{j=0}^{N} w_{n+j}- \frac{\alpha}{\prod_{j=1}^{N-1} w_{n+j}}
\end{equation}
provides a 2-integral of  the modified generalized DTKQ equation (\ref{eqw}).
\end{lemma}
\begin{proof}
A direct calculation shows that 
$$ 
({\cal S}^2 -1) \, \gamma_{n}[w]=
\prod_{j=2}^{N} w_{n+j}\, \Delta \mathcal{E}_n[w] =E_n[w]\, \prod_{j=2}^{N} w_{n+j},  
$$
where 
\begin{equation} \label{Ew}
E_n[w]:= w_{n+N+2}w_{n+N+1}-w_{n+1}w_n 
+\frac{\al}{\prod_{j=2}^N w_{n+j}} \left(\frac{1}{\prod_{k=1}^{N-1}w_{n+k}} - \frac{1}{\prod_{k=3}^{N+1}w_{n+k}} \right) 
\end{equation} 
denotes the lift of (\ref{ushift}) to $N+2$ dimensions obtained from the first substitution in (\ref{wsub}). Hence 
the quantity  $\gam_n[w]$ gives a
2-integral of both (\ref{eqw}) and (\ref{Ew}), where in the first case 
(\ref{gamma}) 
can be rewritten as a function of $w_n,\ldots,w_{n+N-1}$ using (\ref{eqw}).  
\end{proof}

If we substitute for $w_n$ with the ratio of tau functions given in (\ref{wsub}), then we see 
that 
(\ref{gamma}) is equivalent to the bilinear equation (\ref{hir1}). Thus the 
following result is   an immediate consequence of the preceding lemma, and proves the 
first part of Theorem \ref{maintau}. 

\begin{corollary} \label{gamco} 
The quantity given in terms of $\tau_n$ and shifts by 
\beq\label{gamdef} 
\gam_n = \frac{\tau_{n+N+2}\tau_n - \al\tau_{n+N}\tau_{n+2}}{\tau_{n+N+1}\tau_{n+1}} 
\eeq 
is a 2-integral of the multilinear equation (\ref{multi}). 
\end{corollary}  

\begin{remark}\label{oddeven} 
In general, $\gamma_n$ can be considered as a  2-integral for  (\ref{eqw}), (\ref{Ew}) or (\ref{multi}), but it 
can only be considered as a 2-integral for (\ref{equ}) or (\ref{ushift}) when $N$ is odd, because only then can it be 
written purely in terms of $u_n$ and shifts. 
\end{remark} 

As the above remark indicates, there are considerable differences between the cases of even/odd $N$, so henceforth 
we consider these two cases separately. 

\subsection{The even case}

In the case of even $N$, we introduce another dependent variable $v_n$, which is defined by 
\begin{equation} \label{KdVvar}
\pi_2: \quad 
v_n=\prod_{j=0}^{\frac{N-2}{2}} u_{n+2j} 
=\prod_{k=0}^{N-1} w_{n+k} 
=\frac{\tau_n  \tau_{n+N+1}}{\tau_{n+1} \tau_{n+N}}. 
\end{equation}
It turns out that 
$v_n$ satisfies a travelling wave reduction of Hirota's lattice KdV equation, 
\begin{equation} \label{KdV}
V_{k+1,l}-V_{k,l+1}=\alpha\left(\frac{1}{V_{k,l}}-\frac{1}{V_{k+1,l+1}}\right).
\end{equation}
obtained by 
imposing the periodicity condition   
$$V_{k+N,l+1}=V_{k,l}\implies 
V_{k,l}=v_n, \qquad n=l N-k ,
$$
which is called the $(N,1)$-reduction  of (\ref{KdV}). By 
taking the reciprocal of the dependent variable and rescaling, 
this is equivalent to the 
$(N,-1)$-reduction considered in \cite{hkqt}.    
\begin{proposition}
If $N$ is even and $u_n$ is a solution of (\ref{equ}), then 
$v_n= u_nu_{n+2}\cdots u_{n+N-2}$ 
is a solution of 
\begin{equation} \label{RedKdV}
v_{n+{N}+{1}}-v_n=\alpha\left(\frac{1}{v_{n+N}}-\frac{1}{v_{n+1}}\right),
\end{equation}
which is the $(N,1)$ periodic reduction of the 
lattice KdV equation (\ref{KdV}). 
\end{proposition}
\begin{proof} For $N$ even, (\ref{hir1}) is a special case of  
the second bilinear equation in the statement of 
Proposition  4.2 in \cite{hkw}. Upon solving (\ref{hir1}) for 
$\gam_n$ and noting that (since $\gam_n$ has period 2) $({\cal S}^N-1)\gam_n=0$, the equation (\ref{RedKdV}) 
follows immediately by taking $v_n$ to be the ratio of tau functions given in 
(\ref{KdVvar}). 
\end{proof} 

The result of  Proposition 4.2 in \cite{hkw} shows that the $(L,M)$ 
periodic reduction of (\ref{KdV}) is actually associated with two different 
bilinear equations, so applying this result to the case $(L,M)=(N,1)$ considered here, for even $N$ we immediately obtain the second relation (\ref{hir2}) in Theorem \ref{maintau}, in the following form. 

\begin{corollary} For even $N$, the quantity given in terms of $\tau_n$ and shifts by 
\begin{equation} \label{Kiny}
K=\frac{\tau_{n} \tau_{n+2N+1}+\alpha \tau_{n+1} \tau_{n+2N} }{\tau_{n+N}\tau_{n+N+1}}
\end{equation}
is a first integral of (\ref{multi}), which via (\ref{tausub}) produces 
a first integral of (\ref{equ}) or (\ref{ushift}) defined by 
\begin{equation} \label{intku}
K[u]:=\left(\prod_{j=0}^{N-1} u_{n+2j}+\al \right)\, u_{n+N-1}^{\frac{N}{2}}
\prod_{k=1}^{N-2} (u_{n+k} u_{n+2N-2-k })^{
\left\lfloor{\frac{k+1}{2}}\right\rfloor}.
\end{equation}
\end{corollary}

\begin{remark}
Observe that, as it is written, the expression  (\ref{intku}) is a 
polynomial in $ u_{n+j}$ for $j=0,\dots,{2N-2}$, which can be rewritten as 
a rational function of any $N$ adjacent iterates by using the 
recurrence (\ref{equ}) to eliminate higher shifts.   The corresponding first integral (\ref{Kiny}) in terms of tau functions satisfying equation (\ref{multi}) is regarded in a similar way.  By setting 
$u_n=w_nw_{n+1}$, $K[u]$ also provides a first integral for 
 (\ref{eqw}) or (\ref{Ew}). 
\end{remark}

\subsubsection{U-systems and other Liouville integrable maps}

We can now discuss Liouville integrability of various maps associated with (\ref{equ}) for $N$ even. 

Using the presymplectic form 
which comes from   
the cluster algebra associated with the bilinear recurrence (\ref{hir1}) (see \cite{FH} and references), the variables $w_n$ defined in terms of tau functions by (\ref{wsub}) provide a set of symplectic coordinates. 
The corresponding symplectic map in dimension $N$ is defined by
\beq\label{usyseven} 
w_{n+N}w_n \prod_{j=1}^{N-1}w_{n+j}^2 =\gam_n \prod_{j=1}^{N-1}w_{n+j} +\al, 
\eeq  
which is an example of a U-system \cite{hi}, 
and (up to overall scaling) the nondegenerate Poisson bracket preserved
by (\ref{usyseven}) is the one given by equation (3.21) in Lemma 3.13 of \cite{hkqt}, that is 
\beq\label{wbr}
\{w_m,w_n\}=(-1)^{m-n+1}w_m w_n \qquad  \mathrm{for} \qquad  0\leq m < n \leq N-1.
\eeq 
 (For a specific example of this bracket, see equation (\ref{ubr4}) in section \ref{even} below.)

As is shown in \cite{hkqt}, Theorem 3.14, the nondegenerate bracket for (\ref{usyseven}) 
lifts to a bracket for (\ref{RedKdV})  in dimension $N+1$, which has rank $N$ and one Casimir. 
There is another Poisson bracket for (\ref{RedKdV}), coming from a discrete Lagrangian formulation, 
and the two different brackets are compatible with one another.

There is another cluster algebra that arises, namely the one associated with the bilinear  recurrence (\ref{hir2}).  
The variables $v_n$ defined by (\ref{KdVvar}) provide symplectic coordinates for the corresponding U-system, which is the 
map in 
dimension $N$ given by 
\beq\label{uvsys}
\prod_{j=0}^N v_{n+j}=-\al \prod_{k=1}^{N-1}v_{n+k}+K,
 \eeq
 preserving a nondegenerate bracket that has the same form in these coordinates as the one for $w_n$ above, 
i.e. 
\beq\label{vbr} 
\{v_m,v_n\}=(-1)^{m-n+1}v_m v_n \qquad  \mathrm{for} \qquad  0\leq m < n \leq N-1.
\eeq 
This lifts to another Poisson bracket  for the reduced KdV map  (\ref{RedKdV})  in dimension $N+1$, which also  has rank $N$ and one Casimir; in fact, it is a linear combination of the two 
compatible brackets found in  \cite{hkqt}, so they all belong to the same Poisson pencil, consisting of the brackets 
\beq\label{pencil} 
\la_1 \{ \, , \, \}_1+ \la_2 \{ \, , \, \}_2, 
\eeq 
for arbitrary $(\la_1:\la_2)$, 
where $\{\, ,\,\}_{1,2}$ are any two fixed  independent brackets in this family.

According to Corollary 2.2 in \cite{hkw}, each of the bilinear equations (\ref{hir1}) and (\ref{hir2}) has a matrix Lax 
representation ($2\times 2$ and $N\times N$, respectively), and this yields Lax pairs for the corresponding U-systems 
(\ref{usyseven}) and (\ref{uvsys}). However, it is more convenient to make use of the $2\times 2$ Lax pair 
for the discrete KdV reduction, as obtained in \cite{hkqt}. This produces a complete set of 
first integrals for the map  (\ref{RedKdV}), which Poisson commute with respect to any bracket in the pencil (\ref{pencil}).
Hence the Liouville integrability of the map (\ref{ushift}) follows from that of   (\ref{RedKdV}), as described by the following 
result. 

\begin{theorem} \label{maineven} For $N$ even, 
let $\varphi$ and $\chi$ denote the birational maps in dimension $N+1$ defined by (\ref{ushift}) and 
(\ref{RedKdV}) respectively, and let $\psi$ denote the lift of (\ref{usyseven}) to $N+2$ dimensions 
given by 
$$ 
\psi:\qquad (w_0,w_1,\ldots,w_{N-1},\gam_0,\gam_1) \mapsto  (w_1,w_2,\ldots,w_N,\gam_1,\gam_0).
$$ 
Then with $\pi_j$ for $j=1,2$ defined by  (\ref{wsub}) and 
(\ref{KdVvar}),  each of the birational maps $\psi,\varphi,\chi$ preserves a Poisson bracket such that 
the  diagram 
\beq \label{cd}  
\begin{CD}
\C^{N+2} @>\psi >> \C^{N+2}\\ 
@VV\pi_1 V @VV\pi_1 V\\ 
\C^{N+1} @>\varphi >> \C^{N+1} \\
@VV\pi_2 V @VV\pi_2 V\\ 
\C^{N+1} @>\chi>> \C^{N+1} 
\end{CD}
\eeq 
of rational Poisson maps
is commutative. In particular, the bracket preserved by $\varphi$ is of rank $N$, 
being specified  by 
\beq\label{nubr} 
\{ u_0,u_{1}\}=u_0u_{1}, \quad 
 \{u_0, u_{N-1}\}=-\frac{\al}{\prod_{j=1}^{N-2} u_{j}} , 
\quad \{u_0, u_{N}\}= - u_0u_N+ \frac{\al^2}{\left(\prod_{j=1}^{N-1} u_{j}\right)^2} ,
\eeq
with $ \{u_0, u_{j}\}=0$ for $j=2,\ldots, N-2$. Moreover, each of the three horizontal maps is integrable in the Liouville sense. 
\end{theorem} 
\begin{proof} 
As already mentioned,  Theorem 3.14 in \cite{hkqt} says that the bracket 
(\ref{wbr}) lifts to a bracket for (\ref{RedKdV}). This can be made more explicit by first extending (\ref{usyseven}) to the map $\psi$ in dimension $N+2$, which preserves a Poisson bracket of 
rank $N$, defined by extending (\ref{wbr}) to include the extra coordinates $\gam_0,\gam_1$  as a pair of Casimirs. Then the 
formula $v_n = w_nw_{n+1} \cdots w_{n+N-1}$ from 
(\ref{KdVvar}) defines a Poisson map  $\pi: \, \C^{N+2}\to\C^{N+1}$, with a corresponding Poisson bracket for the variables $v_n$, 
$n=0,\ldots,N$, denoted by $\{\, ,\,\}_2$ say (as in \cite{hkqt}), 
so that 
$\pi^* \{ v_m,v_n \}_2  =\{ \pi^*v_m, \pi^*v_n\}$.  
This rational Poisson map factors as $\pi= \pi_2\circ \pi_1$, where 
$\pi_1$ is defined by $u_n=w_nw_{n+1}$, as in (\ref{wsub}), and 
$\pi_2$ by $v_n = u_n u_{n+2} \cdots u_{n+N-2}$. 
A direct calculation shows that the pushforward of the bracket (\ref{wbr}) by $\pi_1$ yields (\ref{nubr}), and the 
product $\gam_0\gam_1$ pushes forward to a Casimir of the latter bracket, which by construction 
is preserved by (\ref{ushift}). By Theorem 4.1 in \cite{hkqt}, the map $\chi$ is Liouville integrable: it has $D+1$ independent first integrals $I_0,I_1,\ldots, I_D$ with $D=\frac{N}{2}$, coming from the trace of a monodromy matrix, which commute with respect to the bracket $\{\, ,\}_2$, with $I_0$ being a Casimir. The pullbacks of these integrals, $\pi_2^*I_j$, provide $D+1$ commuting integrals for the map $\varphi$; and pulling back 
once more gives the same number of integrals $\pi^* I_j$ for $\psi$, including only one Casimir $\pi^*I_0=-(\gam_0\gam_1)^{N/2}$ (cf. Remark 3.15 in \cite{hkqt}), so taking another independent Casimir, i.e.  $\gam_0+\gam_1$, gives a
full set of commuting first integrals. Hence  the Liouville integrability of $\psi$ and $\varphi$ is proved. 
\end{proof} 

\begin{remark} For a fixed value of $K$, the second U-system (\ref{uvsys}) can also be regarded as being Liouville integrable with respect to the nondegenerate bracket (\ref{vbr}). As already mentioned, 
this lifts to another independent bracket in the pencil (\ref{pencil}), 
$\{ \, , \,\}_1$ say,   for the map $\chi$ defined by (\ref{RedKdV}). 
\end{remark}

\subsection{The odd case} 

As is mentioned in Remark \ref{oddeven}, in the case that $N$ is odd, $\gamma_n$ can be considered 
as a 2-integral for (\ref{equ}), and via the substitution (\ref{tausub}) the first bilinear  equation (\ref{hir1}) 
yields 
\beq\label{usysodd} 
\prod_{j=0}^{N-1} u_{n+j} = \gam_n \, \prod_{k=0}^{\frac{N-3}{2}}u_{n+2k+1}+\al, 
\eeq 
which is the U-system associated with this bilinear recurrence.  
The latter defines a symplectic map in dimension $N-1$, whose corresponding nondegenerate Poisson bracket 
is specified by equation (3.22) in Lemma 3.13 of \cite{hkqt}, namely 
\beq\label{uoddbr} 
\{ u_n,u_{n+1}\}=u_nu_{n+1}, \qquad \{u_n,u_{n+j}\}=0 \quad \mathrm{for} \quad 2\leq j\leq N-2.
\eeq 
(For a particular example, 
see (\ref{u5br})  below.) The U-system (\ref{usysodd}) can naturally be viewed as a 
reduction of the Hirota-Miwa equation, and a Lax pair and first integrals can be obtained immediately by applying 
Corollary 2.2 in \cite{hkw}. The case $N=5$ is presented explicitly in section \ref{odd}. 

At this stage,  for the odd case we can already state 
a partial analogue of Theorem \ref{maineven}. 

\begin{theorem} \label{mainodd} For $N$ odd, 
let $\varphi$ denote the birational map in dimension $N+1$ defined by (\ref{ushift}) and   
let $\psi$ denote the birational lift of (\ref{usysodd}) to $N+1$ dimensions 
given by 
$$ 
\psi:\qquad (u_0,u_1,\ldots,u_{N-2},\gam_0,\gam_1) \mapsto  (u_1,u_2,\ldots,u_{N-1},\gam_1,\gam_0).
$$ 
Then there is a birational map $\hat{\pi}_1$ such that 
the  diagram 
\beq \label{cdodd}  
\begin{CD}
\C^{N+1} @>\psi >> \C^{N+1}\\ 
@VV\hat{\pi}_1 V @VV\hat{\pi}_1 V\\ 
\C^{N+1} @>\varphi >> \C^{N+1}
\end{CD}
\eeq 
of birational  Poisson maps
is commutative. Moreover,  the Poisson bracket  preserved by 
$\varphi$ is     of rank $N-1$, with 
non-zero brackets   given by (\ref{nubr}). 
\end{theorem} 
\begin{proof} The bracket (\ref{uoddbr}) for the U-system (\ref{usysodd}) in dimension $N-1$ extends to a bracket for $\psi$ by including 
the additional coordinates $\gam_0,\gam_1$ as two Casimirs. Taking $n=0,1$ in (\ref{usysodd})  defines $u_{N-1}$ and $u_N$ 
as rational functions of $ u_0,u_1,\ldots,u_{N-2},\gam_0,\gam_1$, and conversely gives $\gam_0$ and $\gam_1$ as rational 
functions of $ u_0,u_1,\ldots,u_{N}$, so this specifies a birational transformation $\hat{\pi}_1$ between these two sets of 
coordinates in dimension $N+1$. A direct calculation shows that the bracket preserved by $\varphi$ takes the same form (\ref{nubr}) 
as for $N$ even, but in this case there are two independent Casimirs given by $\gam_0,\gam_1$. 
\end{proof}

When $N$ is odd, the fact that the coefficient $\gam_n$ in (\ref{usysodd}) is 2-periodic means that neither this U-system, nor the corresponding 
bilinear equation (\ref{hir1}), 
can be related to a reduction of discrete KdV, which 
would require the period of $\gam_n$ to divide $N$ (cf. Proposition 3.7 in \cite{hkqt} and 
Proposition 4.2 in \cite{hkw}). However, it turns out that there is a connection with reductions of another integrable two-dimensional lattice equation, namely a discrete form of the Toda lattice.
This connection arises  from the fact that $\tau_n$ satisfies 
the other bilinear equation (\ref{hir3}), 
which is the content of the following statement. 
 
\begin{proposition} For odd $N$, the quantity given in terms of $\tau_n$ and shifts by 
\begin{equation} \label{Kbariny}
{\bar K}=\frac{\tau_{n} \tau_{n+2N+2}-\alpha^2 \tau_{n+2} \tau_{n+2N} }{(\tau_{n+N+1})^2}
\end{equation}
is a first integral of (\ref{multi}), which via (\ref{tausub}) produces 
a first integral of (\ref{equ}) or (\ref{ushift}) defined by 
\begin{equation} \label{intkbaru}
{\bar K}[u]:
=\left(\prod_{j=0}^{N-1} s_{n+2j}-\al^2 \right)\, (s_{n+N-1})^{\frac{N-1}{2}}
\prod_{k=1}^{\frac{N-3}{2}} (s_{n+2k} s_{n+2N-2k-2 })^k
,
\end{equation} 
where
\beq\label{sdef} s_n=u_nu_{n+1}.
\eeq 
\end{proposition}
\begin{proof}Taking $\gam_n$ as given by  (\ref{gamdef}) 
and applying the total difference operator to ${\bar K}$ yields the identity 
$$
\Delta {\bar K}=\frac{\tau_{n+2N+2}\tau_{n+1}}{\tau_{n+N+2}\tau_{n+N+1}} 
({\cal S}^{N+1}-1)\, \gam_{n}-\al \frac{\tau_{n+2N+1}\tau_{n+2}}{\tau_{n+N+2}\tau_{n+N+1}} 
({\cal S}^{N-1}-1)\, \gam_{n+1}, 
$$ 
so that the right-hand side vanishes because $({\cal S}^2-1)\gam_n=0$ and $N$ is odd. 
This completes the proof of the statement, and also the proof of Theorem \ref{maintau}. 
\end{proof}

\begin{remark} The preceding result means that, for each odd $N$, the Somos-$(N+2)$ recurrence (\ref{hir1}) is 
related to (\ref{hir3}), which corresponds to two copies of a Somos-$(N+1)$ recurrence with the iterates interlaced, 
since the iterates with even/odd indices decouple from each other. For the particular case $N=3$, the relation between Somos-5   (with autonomous coefficients) and two copies of 
Somos-4 was shown in Proposition 2.8 of \cite{honetams}, and 
 interpreted as a B\"acklund transformation in \cite{chang}.
\end{remark}

\subsubsection{Lax pair associated with a discrete Toda equation} 

The 
five-point lattice equation 
\begin{equation} \label{Todaeq}
\frac{V_{k,l}}{V_{k+1,l}}-\frac{V_{k-1,l}}{V_{k,l}}+\alpha^2 \left(\frac{V_{k+1,l-1}}{V_{k,l}}-\frac{V_{k,l}}{V_{k-1,l+1}}\right)=0
\end{equation}
is a discrete time Toda equation \cite{Date, book}. 
The $(1,-P)$ periodic reduction of (\ref{Todaeq}) 
corresponds to imposing the condition 
\beq\label{perto} 
V_{k+1,l-P}=V_{k,l} \qquad \implies \qquad 
V_{k,l}=v_n, \qquad n=k P+l,
\eeq 
which leads to 
the ordinary difference equation 
\begin{equation} \label{todared}
\frac{v_n}{v_{n+P}}-\frac{v_{n-P}}{v_n}+\alpha^2 \left(\frac{v_{n+P-1}}{v_{n}}-\frac{v_n}{v_{n+1-P}}\right)=0.
\end{equation}
Upon introducing a tau function $T_n$ such that \beq \label{vtoda}  v_n=\frac{T_n}{T_{n+1}},\eeq  
we can immediately apply 
Proposition 3.1 in  \cite{hkw}, where the $(Q,-P)$ periodic reduction of (\ref{Todaeq}) was considered; here 
we are only concerned with the case $Q=1$, which gives the following  result. 
 
\begin{proposition}\label{todakbar} 
If $v_n$ given by (\ref{vtoda}) satisfies (\ref{todared}), then there is a first integral $\bar{K}$ such that $T_n$ 
satisfies the bilinear equation 
\begin{equation} \label{todaKP}
T_{n+2P}T_{n}=\alpha^2\, T_{n+2P-1}T_{n+1}+ \bar{K} \, T_{n+P}^2,
\end{equation}
and conversely every solution of (\ref{todaKP}) provides a solution of (\ref{todared}). 
\end{proposition}

We now present a Lax representation for \eqref{Todaeq}, which originates from  a map associated with a discretization of the Toda lattice in \cite{Suris1,Suris2}, 
and subsequently provides a Lax representation for  \eqref{todared}.   

\begin{proposition} The discrete Toda equation \eqref{Todaeq} is equivalent to  the 
the discrete zero curvature 
equation 
\begin{equation} \label{laxeqT}
L(\Pi_{k,l},V_{k,l},\eta)M(V_{k+1,l},V_{k,l+1},\eta)=M(V_{k+1,l-1},V_{k,l},\eta)L(\Pi_{k+1,l},V_{k+1,l},\eta),
\end{equation}
where  $\eta$ is a spectral parameter, and 
\begin{equation} \label{LaxPairToda}
L(p,v,\eta)=\left(
\begin{array}{cc}
p+\eta &  v \\
 - v^{-1}
& 0
\end{array}
\right), \ M(u,v,\eta)=\left(
\begin{array}{cc}
1-\alpha^2 uv^{-1} 
-\alpha \eta & -\alpha u \\
  \al v^{-1} 
 & 1
\end{array}
\right). 
\end{equation}
\end{proposition}
\begin{proof} The equation 
\eqref{laxeqT} implies that both 
\beq\label{pdef} \Pi_{k+1,l}= \frac{\alpha V_{k+1,l}}{V_{k,l+1}}+\frac{1}{\alpha}\left(\frac{V_{k,l}}{ V_{k+1,l}}-1\right)\eeq
and 
\beq\label{pdef2} \Pi_{k,l}= \frac{\alpha V_{k+1,l-1}}{V_{k,l}}+\frac{1}{\alpha}\left(\frac{V_{k,l}}{ V_{k+1,l}}-1\right),\eeq
and these two relations together imply the discrete Toda equation \eqref{Todaeq}. \end{proof}

By imposing the periodicity condition (\ref{perto}), the Lax matrices in \eqref{laxeqT}  reduce to 
$$L_n:=L(p_{n},v_{n},\eta), \qquad \ M_n:=M(v_{n},v_{n-P+1},\eta), $$
where from (\ref{pdef}) we have 
\beq\label{pred} 
p_n=\frac{\alpha v_{n}}{v_{n-P+1}}+\frac{1}{\alpha}\left(\frac{v_{n-P}}{ v_{n}}-1\right), 
\eeq 
and the zero curvature equation reduces to 
\begin{equation} \label{laxproof1}
L_nM_{n+P}=M_{n+P-1}L_{n+P}. 
\end{equation}
With this 
notation, we can introduce 
the monodromy matrix as 
\beq\label{monod} \mathcal{M}_{n}:= (1-\al\eta)\, M^{-1}_nL_{n-P+1} \dots L_{n-1} L_n.\eeq  
This satisfies a discrete Lax equation, which follows 
from the identity \eqref{laxproof1}.

\begin{corollary} 
The $(1,-P)$ periodic reduction \eqref{todared} obtained from  the discrete Toda equation  is equivalent to the 
discrete Lax equation 
\begin{equation} \label{LaxRed}
\mathcal{M}_{n}L_{n+1}=L_{n+1} \mathcal{M}_{n+1}.
\end{equation}
\end{corollary}

The equation (\ref{LaxRed}) means that the shift $n\to n+1$ is an isospectral evolution for the 
monodromy matrix (\ref{monod}). From (\ref{LaxPairToda}), the determinant is $\det{\cal M}_n=1-\al\eta$, while 
$\tr{\cal M}_n$ is a monic polynomial of degree $P$ in $\eta$ whose coefficients provide first integrals of  \eqref{todared}.
 
Upon comparing (\ref{Kbariny}) with (\ref{todaKP}), we see that for odd $N$ the solutions of (\ref{equ}) correspond to two interlaced sets of tau functions, 
\beq\label{eo}
T^{\mathrm{even}}_n =\tau_{2n}, \qquad  T^{\mathrm{odd}}_n =\tau_{2n+1},
\eeq 
such that for $P=\frac{N+1}{2}$ there are two sets of solutions of  \eqref{todared} given by 
\beq\label{eot}
  v_n^{\mathrm{even}}=\frac{T_n^{\mathrm{even}}}{T_{n+1}^{\mathrm{even}}}, \qquad  
 v_n^{\mathrm{odd}}=\frac{T_n^{\mathrm{odd}}}{T_{n+1}^{\mathrm{odd}}}.
\eeq 
Then from (\ref{sdef}) we may write the even/odd index quantities $s_j$ as 
$$ 
s_{2n}=\hat{s}^{\mathrm{even}}_n:=\frac{  v_n^{\mathrm{even}}}{ v_{n+1}^{\mathrm{even}}}, \qquad
s_{2n+1}=\hat{s}^{\mathrm{odd}}_n:=\frac{  v_n^{\mathrm{odd}}}{ v_{n+1}^{\mathrm{odd}}}.
$$ 
The equation  \eqref{todared} for  the reduced Toda map is invariant under the scaling $v_n\to \la v_n$, for any 
non-zero $\la$, as are the quantities $\hat{s}^{\mathrm{even/odd}}_n$. Hence in this case \eqref{todared} 
becomes an equation of order $2P-1=N$ for each of the latter quantities, that is 
\beq\label{shat}
\hat{s}_n
\cdots \hat{s}_{n+P-1}-\hat{s}_{n-P}
\cdots \hat{s}_{n-1} 
+\al^2\left(\frac{1}{\hat{s}_n
\cdots \hat{s}_{n+P-2}} - \frac{1}{\hat{s}_{n-P+1}
\cdots \hat{s}_{n-1}}\right)=0 .
\eeq 
Each iteration of (\ref{ushift}) intertwines two sets of solutions of the above equation. 

\begin{proposition} \label{pbrshat} 
For even/odd $n$ taken separately, the formula (\ref{intkbaru}) defines a U-system in dimension $N-1=2P-2$ with coordinates 
$\hat{s}_0,\hat{s}_1, \ldots, \hat{s}_{N-2}$, preserving a nondegenerate Poisson bracket given by 
\beq\label{sbr}  
\{ \hat{s}_n,\hat{s}_{n+1}\} = \hat{s}_n\hat{s}_{n+1}, 
\quad 
\{ \hat{s}_n,\hat{s}_{n+P-1}\} = -2 \hat{s}_n\hat{s}_{n+P-1}, 
\quad 
\{ \hat{s}_n,\hat{s}_{n+P}\} = 2\hat{s}_n\hat{s}_{n+P}, 
\eeq 
with all other brackets $\{ \hat{s}_n,\hat{s}_{n+j}\}$ for $0\leq j\leq N-2$  being zero. This 
 lifts to a bracket of rank $N-1$ in dimension $N$ that is preserved by (\ref{shat}), with $\bar{K}$ being a Casimir, where 
the extra bracket is 
\beq\label{extra} 
\{ \hat{s}_n,\hat{s}_{n+N-1}\} = -\frac{\al^2}{\hat{s}_{n+1}\cdots \hat{s}_{n+N-2}}  
.
\eeq 
\end{proposition} 
\begin{proof}According to Theorem 4.6 in \cite{FH}, the bilinear equation  (\ref{todaKP}) preserves a log-canonical presymplectic form, which reduces to  a symplectic structure for the U-system 
\beq\label{susys}
\hat{s}_{n+P-1}^P\prod_{j=0}^{P-2} (\hat{s}_{n+j}\hat{s}_{n+N-1-j})^{j+1}=\al^2 \hat{s}_{n+P-1}^{P-1}\prod_{k=1}^{P-2} (\hat{s}_{n+k}\hat{s}_{n+N-1-k})^{k}+\bar{K}
\eeq 
in dimension $N-1$.
The symplectic structure is equivalent to a nondegenerate log-canonical Poisson bracket, of the form 
$\{\hat{s}_m,\hat{s}_n\} = c_{mn}\hat{s}_m\hat{s}_n$, where $(c_{mn})$ is a constant 
skew-symmetric Toeplitz matrix. Taking the Poisson bracket of $\hat{s}_{n+j}$ with both sides of (\ref{susys}) for 
$j=1,\ldots, P-2$ produces a system of $2P-4$ homogeneous linear equations for the entries of the first row of this matrix, 
which is readily solved to yield (\ref{sbr}), up to overall scaling by an arbitrary non-zero constant. Upon lifting this bracket to dimension $N$ and requiring that $\bar K$ be a Casimir, taking the bracket of $\hat{s}_n$ with both sides of (\ref{susys}) leads 
to the above expression for $ 
\{ \hat{s}_n,\hat{s}_{n+N-1}\}$. 
\end{proof} 

In the case of odd $N$, a partial  analogue of the bottom part of the diagram (\ref{cd}) arises by taking two iterations of the map 
$\varphi$ in  
(\ref{cdodd}), that is   
\beq\label{cdchi} 
\begin{CD}
\C^{N+1} @>\varphi^2 >> \C^{N+1}\\ 
@VV\hat{\pi}_2 V @VV\hat{\pi}_2 V\\ 
\C^{N} @>\hat{\chi} >> \C^{N}
\end{CD}
\eeq 
where the vertical map $\hat{\pi}_2$ is defined by using  (\ref{sdef}) either for  even values, or for odd values of $n$ only, and  
$$\hat{\chi}: \qquad (\hat{s}_0,\hat{s}_1,\ldots, \hat{s}_{N-1})\mapsto 
 (\hat{s}_1,\hat{s}_2,\ldots, \hat{s}_{N})
$$ is defined by (\ref{shat}) with $P=(N+1)/2$.  
The diagonal entries of the monodromy matrix ${\cal M}_n$ in (\ref{monod}) are functions of the ratios $\hat{s}_j=v_j/v_{j+1}$ 
(although the off-diagonal entries are not), so that $\tr{\cal M}_n$ directly provides first integrals for (\ref{shat}). 
Moreover, from the above diagram, the integrals provided by $\tr{\cal M}_n$ can be 
pulled back by $\hat{\pi}_2$ to give integrals for $\varphi^2$. 

In fact, we can say rather more: the two sets of  integrals obtained by taking 
even/odd $n$ in (\ref{sdef}) coincide, so they pull back to  integrals for $\varphi$. The reason is that the map 
$\varphi$ corresponds to a B\"acklund transformation for the discrete Toda reduction, in the sense of \cite{kuskly}. In order to show this, it is necessary to consider the two sets of reduced Lax matrices  
$$L_n^{\mathrm{even/odd}}:=L(p_{n}^{\mathrm{even/odd}},v_{n}^{\mathrm{even/odd}},\ze), \qquad \ M_n^{\mathrm{even/odd}}:=M(v_{n}^{\mathrm{even/odd}},v_{n-P+1}^{\mathrm{even/odd}},\ze) $$ 
for even/odd $n$  separately,  
and introduce a gauge transformation matrix defined by 
\beq \label{gauge} 
G_n:=\left(\begin{array}{cc} 
\ze +u_{2n+1}-\ka & v_n^{\mathrm{odd}} \\ 
-(v_{n+1}^{\mathrm{even}})^{-1} & -1 \end{array}\right), \qquad \ka = \al^{-1}(1-\gam_0\gam_1).  
\eeq 
\begin{lemma} \label{glem} The gauge matrix (\ref{gauge}) intertwines the even/odd reduced Lax matrices 
as follows: 
\beq\label{inter1} 
L^{\mathrm{even}}_n G_n = G_{n-1} L^{\mathrm{odd}}_n, 
\eeq   
\beq\label{inter2} 
M^{\mathrm{even}}_n G_n = G_{n-P} M^{\mathrm{odd}}_n.
\eeq  
\end{lemma} 
\begin{proof} In order to prove these intertwining relations, it is helpful to note that 
$\det L^{\mathrm{even/odd}}_n =1$ and $\det M^{\mathrm{even/odd}}_n =1-\al \ze$, while   
$\det G_n =\ka -\ze$ follows from the fact that 
$u_{2n+1}=
{v_n^{\mathrm{odd}}}/{v_{n+1}^{\mathrm{even}}}$,  
which is a consequence of the tau function formulae (\ref{eo}) and (\ref{eot}). Thus, in both (\ref{inter1})  and 
(\ref{inter2}), the determinants of the 
left/right-hand sides agree, and henceforth it is sufficient to check only three out of four matrix entries in each equation. 
The $(2,2)$ entries on each side of (\ref{inter1}) are identical, and the same is true for (\ref{inter2}), so we 
need only consider the $(1,2)$ and $(2,1)$ entries. Taking the difference of the $(2,1)$ entries  on each side of   (\ref{inter1})  requires that 
\beq\label{l21} 
-(v_{n}^{\mathrm{even}})^{-1}\Big(u_{2n+1}-\ka -p_n^{\mathrm{odd}}\Big)-(v_{n}^{\mathrm{odd}})^{-1}=0
\eeq 
should hold. By using tau functions it is clear that 
$ 
u_{2n}=
{v_n^{\mathrm{even}}}/{v_{n}^{\mathrm{odd}}}$,  
so that the equality (\ref{l21}) boils down to the identity 
\beq\label{piden} 
p_n^{\mathrm{odd}} = u_{2n}+u_{2n+1}-\ka. 
\eeq 
To prove the latter, we successively use 
ratios of tau functions to go between different variables. On the one hand, we have 
\beq\label{gamid1} 
\gam_0\gam_1 = \Big(v^{\mathrm{even}}_{n-P+1}-\al v^{\mathrm{odd}}_n\Big) 
 \left(\frac{1}{v^{\mathrm{even}}_{n+1}}-\frac{\al}{ v^{\mathrm{odd}}_{n-P+1}}\right),  
\eeq   
which follows from (\ref{hir1}), while on the other hand 
$$  
\frac{v^{\mathrm{odd}}_{n-P}}{v^{\mathrm{odd}}_{n}} 
-\frac{v^{\mathrm{even}}_{n-P+1}}{v^{\mathrm{even}}_{n+1}} 
= u_{2n-N+1}\cdots u_{2n}(u_{2n-N}-u_{2n+1}) =\al (u_{2n}-u_{2n-N+1})
$$
by (\ref{ushift}). Combining the above with the fact that 
$u_{2n-N+1}={v_{n-P+1}^{\mathrm{even}}}/{v_{n-P+1}^{\mathrm{odd}}}$, and then comparing 
(\ref{pdef}) with (\ref{gamid1}), leads to   the alternative formula (\ref{piden}) for 
$p_n^{\mathrm{odd}}$. Similarly, the equality of the $(1,2)$ entries  on each side of  (\ref{inter1}) 
boils down to an equivalent formula for $p_n^{\mathrm{even}}$, that is 
\beq\label{piden2} 
p_n^{\mathrm{even}} = u_{2n-1}+u_{2n}-\ka. 
\eeq 
The equality of the $(2,1)$ entries on each side of (\ref{inter2}) is a direct consequence of the identity (\ref{gamid1}), 
while  to verify    the $(1,2)$ entries it is sufficient to note that 
\beq\label{gamid2} 
\gam_0\gam_1 = \Big(v^{\mathrm{odd}}_{n-P}-\al v^{\mathrm{even}}_n\Big) 
 \left(\frac{1}{v^{\mathrm{odd}}_{n}}-\frac{\al}{ v^{\mathrm{even}}_{n-P+1}}\right).  
\eeq   
Both (\ref{gamid1}) and (\ref{gamid2}) are proved in the same way, by expressing the terms on their  
right-hand sides in terms of tau functions, and using (\ref{hir1}) together with the 2-periodicity of $\gam_n$. 
\end{proof} 

For the case of odd $N$, we can now state a closer analogue of the bottom part of the diagram (\ref{cd}). 
\begin{theorem}\label{todatheorem} 
For odd $N=2P-1$, the map $\varphi$ defined by (\ref{ushift}) corresponds to a B\"acklund transformation (BT)  
for the reduced Toda equation (\ref{todared}). The BT is a 2-valued Poisson correspondence between solutions of  
(\ref{shat}), which preserves all the first  integrals obtained from the trace of the monodromy matrix (\ref{monod}), and there 
is a commutative diagram 
\beq\label{cdchibt} 
\begin{CD}
\C^{N+1} @>\varphi >> \C^{N+1}\\ 
@VV\hat{\pi}_2 V @VV\hat{\pi}_2 V\\ 
\C^{N} @>{\chi}_{BT} >> \C^{N}
\end{CD}
\eeq 
where ${\chi}_{BT}$ denotes one of the branches of the correspondence. 
\end{theorem}  
\begin{proof} To begin with, suppose that $u_n$ is a solution of (\ref{ushift}), with $\tau_n$ being a corresponding 
tau function, so that $u_n =\tau_{n+3}\tau_n/(\tau_{n+2}\tau_{n+1})$. Then for the gauge matrix $G_n$ given by 
(\ref{gauge}), repeated application of  (\ref{inter1}) shows that 
$$ 
\begin{array}{rcl}
(M_n^{\mathrm{even}})^{-1}L_{n-P+1}^{\mathrm{even}} \dots L_{n-1}^{\mathrm{even}} L_n^{\mathrm{even}}G_n 
& = & (M_n^{\mathrm{even}})^{-1}L_{n-P+1}^{\mathrm{even}} \dots L_{n-1}^{\mathrm{even}} G_{n-1} L_n^{\mathrm{odd}} \\  
& = & \cdots  \\ 
& = & (M_n^{\mathrm{even}})^{-1} G_{n-P}L_{n-P+1}^{\mathrm{odd}} \dots L_{n-1}^{\mathrm{odd}} L_n^{\mathrm{odd}}, 
\end{array} 
$$ 
and then by applying (\ref{inter2}) it follows that the monodromy matrices for the even/odd index solutions of   
(\ref{todared}) are related 
by 
\beq\label{btg} 
\mathcal{M}_{n}^{\mathrm{even}}G_n = G_n \mathcal{M}_{n}^{\mathrm{odd}}, 
\eeq 
proving the claim that the first integrals for these two sets of solutions coincide. 
 
Now suppose instead that an adjacent  set of $2P$ variables $v_j^{\mathrm{even}}$ 
 is given, corresponding to a set of initial data for (\ref{todared}), and define a 
transformation to another set $v_j^{\mathrm{odd}}$, say with $j=0,\ldots, 2P-1$ in each case, by the gauge 
transformation of monodromy matrices (\ref{btg}) for $n=2P-1$. With this choice of indices,  $v_{2P}^{\mathrm{even}}$ 
appearing in $G_{2P-1}$ should be specified in terms of $v_j^{\mathrm{even}}$  for $0\leq j\leq 2P-1$ 
according to (\ref{todared}), 
and the quantities $u_k$ with even/odd  indices can be defined as the ratios 
$$ 
u_{2j}=\frac{v_j^{\mathrm{even}} }{ v_{j}^{\mathrm{odd}}},  \qquad  
u_{2j+1}=\frac{v_j^{\mathrm{odd}}}{ v_{j+1}^{\mathrm{even}}}, \qquad j=0,\ldots, 2P-1.
$$ 
Imposing the condition (\ref{inter1}) for $n=P,\ldots, 2P-1$ implies that $v_j^{\mathrm{odd}}$ for 
$j=0,\ldots, 2P-1$ are determined completely by the initial $v_j^{\mathrm{even}}$, together with $u_{4P-1}$ and 
the B\"acklund parameter $\ka$ appearing in each of the $G_n$. The requirement that (\ref{inter2}) should also hold then ensures 
that (\ref{btg}) is satisfied for $n=2P-1$,  that is 
\beq \label{spec} 
\mathcal{M}_{2P-1}^{\mathrm{even}}(\ze)G_{2P-1} (\ze) = G_{2P-1} (\ze) \mathcal{M}_{2P-1}^{\mathrm{odd}}(\ze ) 
\eeq 
for all $\ze $, so the monodromy matrices    $\mathcal{M}_{2P-1}^{\mathrm{even/odd}}$ 
have the same spectrum, and this  requirement imposes an additional relation on $u_{4P-1}$. In that case,  the correspondence between 
 $v_j^{\mathrm{even/odd}}$ is fixed up to a choice of square root, which can be seen directly by applying 
the notion of spectrality from \cite{kuskly}: upon  noting that, for  $\ze =\ka$,  a vector in the kernel of the transposed gauge matrix   is given by  
$$ {\bf v}= \left(
\begin{array}{c}( v_{2P}^{\mathrm{even}})^{-1} \\ u_{4P-1} \end{array}\right) \implies  G_{2P-1}^T(\ka ) {\bf v} =\mathbf{0}, 
$$  
it follows that  
$$ 
 {\bf v}^T \mathcal{M}_{2P-1}^{\mathrm{even}}(\ka)G_{2P-1} (\ka) =\mathbf{0}^T
\implies  {\bf v}^T \mathcal{M}_{2P-1}^{\mathrm{even}}(\ka)=\mu \,  {\bf v}^T 
$$ 
for some $\mu$, 
so $ {\bf v}^T$ is a left eigenvector of the monodromy matrix for this value of $\ze$. Thus $(\ka , \mu)$ is a point 
on the spectral curve 
$$ 
\mu^2 - \tr  \mathcal{M}_{2P-1}^{\mathrm{even}}(\ka)\, \mu + 1-\al \ka  
=0.
$$ 
So, for a fixed choice of the initial data and the parameter $\ka$, there are two possible values of $\mu$, 
and by writing 
$$ \mathcal{M}_{2P-1}^{\mathrm{even}}(\ka) = \left(
\begin{array}{cc}a & b  \\ c & d  \end{array}\right)
$$ 
this leads to  
two possible values for 
$$ 
u_{4P-1}=\frac{\mu-a}{v_{2P}^{\mathrm{even}}c} 
= \frac{b}{v_{2P}^{\mathrm{even}}(\mu-d)}.
$$
Hence the BT defined in this way is a 2-valued correspondence between $v_j^{\mathrm{even/odd}}$, and 
also provides a 2-valued correspondence between the quantities $\hat{s}_j^{\mathrm{even/odd}}$ for $j=0,\ldots, N-1$.  The iteration of the map $\varphi$ defined by (\ref{ushift}) corresponds to one particular branch of the correspondence, $\chi_{BT}$ say, with the 
other branch corresponding to $\varphi^{-1}$.

It remains to verify that this is a Poisson correspondence, preserving the Poisson structure  for the coordinates $\hat{s}_j$ in Proposition \ref{pbrshat}. To see this, note that  the branch $\chi_{BT}$ is associated with the bilinear equation (\ref{hir3}), 
which arises from a cluster algebra, and takes the form of two copies of  (\ref{todaKP}) for even/odd indices. By applying 
Theorem 4.6 in \cite{FH}, the corresponding presymplectic form reduces to a symplectic form $\hat\omega$ in 
dimension $2N-2$ for the combined U-system 
coordinates  $\hat{s}_j^{\mathrm{even/odd}}$ for $j=0,\ldots, N-2$, being a sum of two identical  
symplectic forms  which are switched under the action of the map $\varphi$, that is 
$$ 
\hat\om = \hat\om^{\mathrm{even}}+\hat\om^{\mathrm{odd}}, \qquad 
\varphi^* \hat\om^{\mathrm{even/odd}} = \hat\om^{\mathrm{odd/even}}. 
$$
Now $\chi_{BT}$  preserves all the first integrals of (\ref{shat}), including $\bar K$, hence it must also preserve 
the lifted bracket in dimension $N$ given by (\ref{sbr}) and (\ref{extra}). 
\end{proof} 

The Liouville integrability of 
the map $\hat\chi$ in (\ref{cdchi}), defined by (\ref{shat}),   is worthy of a more detailed treatment elsewhere, as is the 
connection of the BT with that for the even Mumford systems in \cite{kuvan}.  
In section \ref{odd} below we merely present the details for the particular case 
$N=5$.

\section{An even example: $N=4$}\label{even} 

For  $N=4$ the equation \eqref{equ} becomes 
\begin{equation} \label{EqN4}
(u_n+u_{n+1}+u_{n+2}+u_{n+3}+u_{n+4}+\beta)u_{n+1}u_{n+2}u_{n+3}=\alpha, 
\end{equation}
and its lift (\ref{ushift})  is the map $\varphi$ in  dimension 5 defined by  
\beq\label{sh4} 
\varphi: \qquad 
u_{n+5}-u_n + \frac{\al}{u_{n+2}u_{n+3}}\left( \frac{1}{u_{n+1}}- \frac{1}{u_{n+4}}\right)=0. 
\eeq 
If we set $u_n=\frac{\tau_{n+3} \tau_n}{\tau_{n+2} \tau_{n+1}}$, then the 
tau function $\tau_n$ satisfies (\ref{multi}), which in this case is of degree 6, being given by 
\beq\label{multi4} 
\begin{array} {rcl} 
\tau_{7}\tau_{4}^2\tau_{3}\tau_{2} \tau_{1} & = &   
\al \tau_{5}\tau_{4}^2\tau_{3}^2\tau_{2}  
-\be 
\tau_{6} \tau_{5}  \tau_{4} \tau_{3} \tau_{2} \tau_{1}  
 -  \tau_{6} \tau_{5}  \tau_{4} \tau_{3}^2  \tau_{0} 
 -  \tau_{6} \tau_{5}  \tau_{4}^2 \tau_{1}^2  \\
&& 
 -  \tau_{6} \tau_{5}^2  \tau_{2}^2 \tau_{1} 
 -  \tau_{6}^2 \tau_{3}^2  \tau_{2} \tau_{1} 
\end{array}  \eeq 
(with $n\to 0$ for brevity).
The iterates of (\ref{multi4}) satisfy a Somos-6 relation, namely the  first bilinear equation 
(\ref{hir1}), which  
takes the form 
 \begin{equation} \label{N4bil1}
 \tau_{n+6} \tau_{n}=\gamma_n \tau_{n+5} \tau_{n+1}+\alpha \tau_{n+4} \tau_{n+2},
\qquad \gamma_{n+2}=\gamma_n, 
 \end{equation}
while the second bilinear equation (\ref{hir2})  is 
\begin{equation} \label{N4bil2}
 \tau_{n+9} \tau_{n}=K \tau_{n+5} \tau_{n+4}-\alpha \tau_{n+8} \tau_{n+1},
 \end{equation}
in this case. 

For $N=4$, taking  $w_n=\frac{\tau_{n} \tau_{n+2}}{ \tau_{n+1}^2}$ yields 
the U-system (\ref{usyseven}) for 
(\ref{N4bil1}). Each iteration 
of the U-system is symplectic, and it  lifts to the map %
\beq\label{usys4} 
\psi: \quad 
(w_0,w_1,w_2,w_3,\gam_0,\gam_1)\mapsto 
\left( w_1,w_2,w_3,
\frac{\gam_0 w_{1}w_{2}w_{3}+\al }{w_0 w_{1}^2w_{2}^2w_{3}^2} 
,\gam_1,\gam_0\right)
\eeq  
in six dimensions, 
preserving the 
log-canonical Poisson bracket given by 
\beq\label{ubr4}
\{w_m,w_n\}=c_{mn}w_mw_n, \quad (c_{mn})_{0\leq m,n\leq 3} =\left(\begin{array}{cccc} 
0 & 1 & -1 & 1 \\ 
-1 & 0 & 1 & -1 \\ 
1 & -1 & 0 & 1 \\ 
-1 & 1 & -1 & 0 \end{array}\right), \quad  \{\gam_m,w_n\}=0.
\eeq 
Under the map defined by setting $u_n=w_nw_{n+1}$, that is 
$$ 
\pi_1: \qquad 
(w_0,w_1,w_2,w_3,\gam_0,\gam_1)\mapsto (w_0w_1,w_1w_2,w_2w_3,w_3w_4,w_4w_5)
$$ 
where 
$$
w_4=\psi^*w_3=\frac{\gam_0 w_{1}w_{2}w_{3}+\al }{w_0 w_{1}^2w_{2}^2w_{3}^2} , \qquad
w_5=\psi^*w_4=\frac{\gam_1 w_{2}w_{3}w_{4}+\al }{w_1 w_{2}^2w_{3}^2w_{4}^2},   
$$
the bracket (\ref{ubr4}) yields the bracket (\ref{nubr}) preserved by (\ref{sh4}).

The second U-system (\ref{uvsys}), associated with  \eqref{N4bil2}, is obtained by taking 
$v_n=\frac{\tau_n  \tau_{n+5}}{\tau_{n+1} \tau_{n+4}}=u_n u_{n+2},$ 
producing the  birational map 
\begin{equation} \label{sympmap4}
\hat{\psi}: \qquad (v_0,v_1,v_2,v_3)\mapsto\left(v_1,v_2,v_3,\frac{ K-\alpha v_1 v_2 v_3}{v_0 v_1 v_2 v_3}\right), 
\end{equation} 
which is symplectic with respect to the 2-form 
$$\omega=\sum_{0 \le i<j \le 3}\frac{1}{v_i v_j}\rd v_i \wedge \rd v_j.$$ Up to overall scale, this 
symplectic form corresponds to the nondegenerate 
log-canonical Poisson bracket given by 
$\{v_m,v_n\}=c_{mn}v_mv_n$, with the same coefficients $c_{mn}$ as in (\ref{ubr4}).

The $(4,1)$ periodic reduction of the 
lattice KdV equation, given by setting $N=4$ in (\ref{RedKdV}), 
 is equivalent to the $5$-dimensional birational map
\begin{equation} \label{mapkdv}
\chi: \qquad (v_0,v_1,v_2,v_3,v_4)
\mapsto 
\left(v_1,v_2,v_3,v_4,v_0+\alpha\Big(\frac{1}{v_4}-\frac{1}{v_1}\Big)\right).
\end{equation} 
This arises either by lifting (\ref{sympmap4}) to one dimension higher and  eliminating $K$, which becomes a
 first integral for (\ref{mapkdv}) in the form 
\beq\label{kform} 
K=v_0v_1v_2v_3v_4+\al v_1v_2v_3, 
\eeq or by using $v_n=u_nu_{n+2}$ to obtain the transformation 
$$ 
\pi_2: \qquad (u_0,u_1,u_2,u_3,u_4) 
\mapsto (u_0u_2,u_1u_3,u_2u_4,u_3u_5,u_4u_6). 
$$
In the first case, the nondegenerate bracket for (\ref{sympmap4}) lifts to  the bracket $\{\,,\}_1$ defined by 
\beq\label{4br1}
\{v_0,v_1\}_1 = v_0v_1,\,  \{v_0,v_2\}_1 = -v_0v_2, \, \{v_0,v_3\}_1 = v_0v_3,\, \{v_0,v_4\}_1 = -v_0v_4-\al,  
\eeq 
while the bracket (\ref{nubr}) is pushed forward by $\pi_2$ to the bracket $\{\,,\}_2$ specified by 
\beq\label{4br2}
\begin{array}{ll} 
\{v_0,v_1\}_2 = v_0v_1-\al, & \{v_0,v_2\}_2 = -v_0v_2+{\al^2}{v_1^{-2}},  \\ 
\{v_0,v_3\}_2 = v_0v_3-{\al^3}{(v_1v_2)^{-2}}, & \{v_0,v_4\}_2 = -v_0v_4+{\al^4}{(v_1v_2v_3)^{-2}}. 
\end{array} 
\eeq 
The Poisson brackets $\{\,,\}_{1,2}$ are compatible with each other, and both are preserved by 
(\ref{mapkdv}).

From the Lax representation of the KdV equation we derive the corresponding monodromy matrix for the 
$(4,1)$-reduction, as in  \cite{hkqt}, that is 
$$\mathcal{M}(v_0,v_1,v_2,v_3,v_4, \la )=M(v_4, \la )L(v_3,v_4, \la )L(v_2,v_3, \la )L(v_1,v_2, \la )L(v_0,v_1,\la ),$$
where  $\la$ is a  spectral parameter, and 
\begin{equation} \label{LaxPair2}
L(V,W,\lambda)=\left(
\begin{array}{cc}
V-\frac{\alpha}{W} & \ \lambda \\
 1 & 0
\end{array}
\right), \qquad M(V,\lambda)=\left(
\begin{array}{cc}
V & \ \lambda \\
 1 & \frac{\alpha}{V}
\end{array}
\right).
\end{equation}
The associated discrete Lax  equation for the  map  \eqref{mapkdv} is 
$$L(v_0,v_1,\lambda)\mathcal{M}(v_0,v_1,v_2,v_3,v_4,\la ) =\mathcal{M}(v_1,v_2,v_3,v_4,v_5,\la )
L(v_0,v_1,\lambda), $$
and the trace of the monodromy matrix is given by 
$$\tr \mathcal{M}(v_0,v_1,v_2,v_3,v_4,\la )=I_2 \lambda^2+ I_1 \lambda +I_0,$$
where the coefficients are  three functionally independent integrals, namely  
\begin{eqnarray*}
 I_0 &=& g_0 g_1 g_2 g_3 g_4 , \\ 
I_1 &=&  g_0 g_1 g_2+g_1 g_2 g_3+g_0 g_1 g_4+g_0 g_3 g_4+g_2 g_3 g_4 + \frac{\alpha g_2 g_3}{g_0}, \\ 
I_2 &=& g_0+g_1+g_2+g_3+g_4+\frac{\alpha}{g_0},
\end{eqnarray*}
conveniently expressed in terms of the quantities $g_0=v_0$  and $ g_i=v_i-  \alpha /v_{i-1}$ for $i=1,2,3,4$. 
Comparison with (\ref{kform}) reveals that $K$, 
a Casimir for the bracket $\{\, ,\,\}_1$,  is expressed as 
$$ 
K=I_2 \al^2+ I_1 \al +I_0, 
$$
while $I_0$ is a Casimir for $\{\, ,\,\}_2$, and all of these integrals are 
in involution with respect to both brackets. 

By setting $v_n=u_n u_{n+2}$, the $I_j$ pull back to three integrals for  the map (\ref{sh4}), which commute
with respect to the bracket defined by (\ref{nubr}) with $N=4$. A further pullback provides three independent commuting integrals 
for (\ref{usys4}), with a fourth one being the Casimir $\gam_0+\gam_1$.

\section{An odd example: $N=5$} \label{odd}

For $N=5$ the equation (\ref{equ}) is 
\beq\label{dtkq5} 
(u_n+u_{n+1}+u_{n+2}+u_{n+3}+u_{n+4}+u_{n+5}+\be )u_{n+1}u_{n+2}u_{n+3}u_{n+4}=\al, 
\eeq  
which via (\ref{tausub})  corresponds  to  
the degree 7 equation (\ref{multi5}), whose iterates also satisfy %
a Somos-7 recurrence with a period 2 coefficient, given by 
\beq\label{s7} 
\tau_{n+7}\tau_n =\gam_n\,\tau_{n+6}\tau_{n+1}+\al \tau_{n+5}\tau_{n+2}. 
\eeq  
The  U-system associated with (\ref{s7})  is 
\beq\label{u5} 
u_nu_{n+1}u_{n+2}u_{n+3}u_{n+4}=\gam_n u_{n+1}u_{n+3}+\al, \qquad \gam_n=\gam_{n+2}, 
\eeq 
and 
the nondegenerate log-canonical Poisson bracket in 4 dimensions for 
(\ref{u5}) is given by 
\beq\label{u5br} 
\{u_n,u_{n+1}\}=u_nu_{n+1}, \qquad 
\{u_n,u_{n+2}\}=0 = \{u_n,u_{n+3}\}. 
\eeq 

By eliminating $\beta$ from (\ref{dtkq5}), or eliminating $\gam_n$ from (\ref{u5}), we obtain a lift to the same 
equation in 6 dimensions, namely the $N=5$ case of (\ref{ushift}), which is equivalent to 
\beq\label{lift6} 
u_{n+6}-u_n = \frac{\al}{u_{n+2}u_{n+3}u_{n+4}}\,\left(\frac{1}{u_{n+5}}-\frac{1}{u_{n+1}}\right). 
\eeq 
Upon taking the bracket of both sides of (\ref{u5}) with $u_0$ for $n=0,1$, we see that (\ref{u5br}) 
lifts to a Poisson bracket of rank 4  in 6 dimensions,  with the additional brackets being  
\beq\label{addu5br} 
\{u_n,u_{n+4}\}=-\frac{\al}{u_{n+1}u_{n+2}u_{n+3}}, 
\qquad 
\{u_n,u_{n+5}\}=-u_nu_{n+5}+ \frac{\al^2}{u_{n+1}^2u_{n+2}^2u_{n+3}^2u_{n+4}^2}. 
\eeq 
This 6-dimensional 
bracket is preserved by (\ref{lift6}).

From Proposition 2.1 and Corollary 2.2 in \cite{hkw},  the bilinear equation (\ref{s7}) is the compatibility condition of the scalar Lax pair 
$$ 
Y_n\phi_{n+6} +\al\zet \phi_{n+4}=\xi\phi_n, \qquad \phi_{n+2}=\frac{1}{u_{n+1}}\Big(-\zet\phi_n +\phi_{n+1}\Big), 
$$ 
where $u_n$ is given in terms of the tau function  by (\ref{tausub}), 
$\zet$ and $\xi$ are  spectral parameters, and 
$$ 
Y_n =\frac{\tau_{n+8}\tau_n}{\tau_{n+6}\tau_{n+2}}=u_nu_{n+1}u_{n+2}u_{n+3}u_{n+4}u_{n+5}. 
$$ 
For $n=0$ the scalar Lax pair can be rewritten as a $2\times 2$ matrix system in terms of 
$u_0,u_1,\ldots, u_5$, leading directly to a Lax pair for the map 
$$\varphi: \qquad (u_0,\ldots,u_5)\mapsto (u_1,\ldots,u_6)$$ 
corresponding to (\ref{lift6}), given by 
\beq\label{lax6} 
{\bf L}(\zet )\Phi = \xi \Phi, \qquad  \tilde{\Phi} ={\bf M}(\zet) \Phi, 
\eeq 
with  the tilde denoting the index shift $n\to n+1$, and 
$$ 
{\bf L} (\zet) =\sum_{j=0}^3{\bf L}^{(j)}\zet^j, \qquad {\bf M}(\zet) = \left(\begin{array}{cc} 0 & 1 \\ 
-\frac{\zet}{u_1} & \frac{1}{u_1} \end{array} \right), 
$$ 
where
$$ 
{\bf L}^{(0)}  = \left(\begin{array}{cc} 0 & u_0 \\ 
0 & 1 \end{array} \right), \qquad 
{\bf L}^{(1)}  = \left(\begin{array}{cc}-u_0 &-\frac{\gam_0}{u_2}- u_0(u_1+u_2+u_3) \\ 
-1 & \beta+u_0 \end{array} \right),
$$ 
$$ 
{\bf L}^{(2)}  = \left(\begin{array}{cc}  \frac{\gam_0}{u_2}+u_0(u_2+u_3) & \gam_0+u_0u_1(u_3+u_4)-\frac{\al}{u_2u_3}  \\ 
-\be -u_0-u_1 & \gam_1\left(\frac{1}{u_1}+\frac{1}{u_3} \right) +u_1(u_3+u_4)+u_2(u_4+u_5)-\frac{\al}{u_1u_3u_4} \end{array} \right), 
$$ 
$$ 
{\bf L}^{(3)}  = \left(\begin{array}{cc}-\gam_0 & 0 \\ 
-\frac{\gam_1}{u_1}-u_2(u_4+u_5)+\frac{\al}{u_1u_3u_4}  & -\gam_1 \end{array} \right).
$$ 
In the above formulae, $\be,\gam_0,\gam_1$ stand for the functions of $u_j$ defined by (\ref{dtkq5}) for $n=0$, and by  
(\ref{u5}) for $n=0,1$, respectively.   
 
The compatibility condition for    the system (\ref{lax6}) is the discrete Lax equation 
$$ 
\tilde{{\bf L}}(\zet ){\bf M}(\zet )={\bf M}(\zet ){\bf L}(\zet ). 
$$  
The spectral curve corresponding to  the Lax matrix ${\bf L}(\zet )$  is a  curve of genus 2  in the $(\zet,\xi)$ plane, 
\beq\label{spec5} 
\det ({\bf L} (\zet)-\xi\mathbf{1})\equiv \xi^2 
+(K_3\zet^3-K_2\zet^2+K_1\zet -1)\xi +K_0\zet^6+\al\zet^5
=0, 
\eeq
whose coefficients $K_j$ provide 4 functionally independent first integrals for (\ref{lift6}), namely 
$$ 
K_0=u_0u_1u_2u_3u_4u_5-\al (u_0+u_5)+\frac{\al^2}{u_1u_2u_3u_4}, 
\, K_1= u_0+u_{1}+u_{2}+u_{3}+u_{4}+u_{5}- \frac{\al}{u_1u_2u_3u_4}, 
$$ 
$$ 
K_2= \sum_{j=0}^5 u_ju_{j+2}+\sum_{j=0}^2 u_ju_{j+3}-\al\,\left(\frac{1}{u_1u_2u_3}+ 
\frac{1}{u_1u_2u_4}+\frac{1}{u_1u_3u_4}+\frac{1}{u_2u_3u_4}\right), 
$$ 
$$  
K_3=u_0u_2u_4+u_1u_3u_5 -\al\,\left(\frac{1}{u_1u_3}+\frac{1}{u_2u_4}\right), 
$$ 
with indices read $\bmod \,6$ in the first sum above. 
The first integral in (\ref{zehat}) is 
 $$\hat{\zeta}=u_0 u_1 u_2 u_3 u_4 u_5+\alpha (u_1+u_2+u_3+u_4)= K_0+\al K_1.$$
From (\ref{dtkq5}) and (\ref{u5}) we 
can identify 
$$ 
K_0=\gam_0\gam_1, \qquad  K_3=\gam_0+\gam_1, 
\qquad 
 K_1=-\be= 
u_0+u_1+u_2+u_3+\frac{\gam_0}{u_0u_2}+\frac{\gam_1}{u_1u_3}+\frac{\al}{u_0u_1u_2u_3} ,  
$$ 
$$ 
K_2 = 
u_0u_2+u_0u_3+u_1u_3+\frac{\gam_0(u_0+u_1+u_2)}{u_0u_2}+ \frac{\gam_1(u_1+u_2+u_3)}{u_1u_3}
+\al\left(\frac{1}{u_0u_1u_3}+\frac{1}{u_0u_2u_3}\right). 
$$
By construction, if we consider $\gam_0,\gam_1$ as functions of $u_j$ defined by (\ref{u5}) for $n=0,1$, then these are 
Casimirs of the bracket given by (\ref{u5br}) and (\ref{addu5br}). Hence $K_0,K_3$ are also Casimirs of this bracket, and one can verify directly that $\{K_1,K_2\}=0$, which shows that  (\ref{lift6}) is a Liouville integrable map in 6 dimensions.   

For $N=5$, the other bilinear equation in Theorem \ref{maintau} is (\ref{hir3}), which in this case becomes
\beq\label{todatau} 
\tau_{n+12}\tau_n =\alpha^2 \tau_{n+10}\tau_{n+2} +\bar{K}\tau_{n+6}^2. 
\eeq  
From the substitution (\ref{tausub}), the conserved quantity $\bar K$ in (\ref{todatau}) can be 
written %
in terms of $u_j$ for $0\leq j\leq 9$, as defined in (\ref{intkbaru}), which gives 
$$
\bar{K} = u_0u_1u_2^2u_3^2u_4^3u_5^3u_6^2u_7^2u_8u_9-\al^2 
u_2u_3u_4^2u_5^2u_6u_7, 
$$
and then using (\ref{u5}) the resulting expression can be further rewritten 
as a function of $\gamma_j$ and only four adjacent $u_j$, which reveals that 
it is a polynomial in the 
quantities $K_j$ obtained from the Lax  pair above, that is 
\beq\label{kbarid} 
\bar{K} = K_0^3+\al K_0^2K_1+\al^2 K_0 K_2 +\al^3 K_3. 
\eeq  

For $P=3$, the recurrence \eqref{todared}  corresponds to the 
six-dimensional map 
\beq\label{todam} 
(v_0,v_1,\dots,v_5)\mapsto \left( v_1,v_2,\dots, \frac{v_1 v_3^2}{v_0 v_1+\alpha^2(v_3^2-v_1 v_5)}\right), 
\eeq  
and 
in this case the monodromy matrix (\ref{monod})  is  
 \begin{equation*}
 \mathcal{M}_{5} (\ze )= (1-\alpha \ze) M(v_5,v_{3},\ze)^{-1}L(p_{3},v_{3},\ze)L(p_{4},v_{4},\ze) L(p_5,v_5,\ze), 
 \end{equation*}
where, from (\ref{pred}),  there is dependence on $v_0,v_1,v_2$  via 
$p_3=\alpha v_3/{v_{1}}+{\alpha}^{-1}({v_{0}}/{ v_{3}}-1)$, and similarly for $p_4,p_5$. By taking the ratios 
$\hat{s}_n = v_n/v_{n+1}$, the iterates of  (\ref{todam}) can be reduced to those 
of  (\ref{shat}), which in this case yields the map 
\beq\label{shat5} 
\hat{\chi} : \qquad 
(\hat{s}_0, \hat{s}_1,\hat{s}_2,\hat{s}_3,\hat{s}_4)\mapsto 
\left(\hat{s}_1,\hat{s}_2,\hat{s}_3,\hat{s}_4, 
\frac{\hat{s}_0\hat{s}_1^2\hat{s}_2^2\hat{s}_3\hat{s}_4+\al^2 (\hat{s}_3\hat{s}_4-\hat{s}_1\hat{s}_2)}
{\hat{s}_1\hat{s}_2\hat{s}_3^2\hat{s}_4^2} 
\right) , 
\eeq  
and by Proposition \ref{pbrshat} this preserves the Poisson bracket in 5 dimensions given by 
\beq\label{s5br}\begin{array}{ll} 
\{ \hat{s}_n,\hat{s}_{n+1} \} = \hat{s}_n\hat{s}_{n+1}, \, &  
\{ \hat{s}_n,\hat{s}_{n+2} \} = -2\hat{s}_n\hat{s}_{n+2}, \\
\{ \hat{s}_n,\hat{s}_{n+3} \} = 2\hat{s}_n\hat{s}_{n+3}, & 
\{ \hat{s}_n,\hat{s}_{n+4} \} = -{\al^2}({\hat{s}_{n+1}\hat{s}_{n+2}\hat{s}_{n+3}})^{-1}.
\end{array}    
\eeq 
By construction, the above bracket has 
$$\bar{K}=\hat{s}_1\hat{s}_2^2 \hat{s}_3(\hat{s}_0 \hat{s}_1\hat{s}_2\hat{s}_3 \hat{s}_4-\alpha^2)$$ 
as a Casimir. 
The trace of the monodromy matrix 
is a monic cubic polynomial in $\eta$, 
 $$\tr  \mathcal{M}_{5} (\ze )= \ze^3+H_2 \ze^2+H_1\ze+H_0, $$
 where $H_0$, $H_1$, $H_2$ provide three functionally independent first integrals for the map (\ref{todam}), 
but since they depend only on the ratios $v_n/v_{n+1}$ they are also first integrals for (\ref{shat5}), with the 
explicit expressions 
\begin{eqnarray*}
{H}_2 &=& \frac{1}{\alpha}(\hat{s}_0 \hat{s}_1 \hat{s}_2+\hat{s}_1 \hat{s}_2 \hat{s}_3+\hat{s}_2 \hat{s}_3 \hat{s}_4-3)
+\alpha\left(\frac{1}{\hat{s}_1 \hat{s}_2}+\frac{1}{\hat{s}_2 \hat{s}_3}\right), \\ 
{H}_1 &=& \frac{1}{\alpha^2}(\hat{s}_0 \hat{s}_1^2 \hat{s}_2^2 \hat{s}_3+\hat{s}_0 \hat{s}_1 \hat{s}_2^2 \hat{s}_3 \hat{s}_4+
\hat{s}_1 \hat{s}_2^2 \hat{s}_3^2 \hat{s}_4-3) 
+\frac{\hat{s}_3 \hat{s}_4}{\hat{s}_1}+\frac{\hat{s}_0 \hat{s}_1}{\hat{s}_3}-\hat{s}_2
+\frac{\alpha^2}{\hat{s}_1 \hat{s}_2^2 \hat{s}_3}-\frac{2H_2}{\alpha}. 
\end{eqnarray*}
The formula for $H_0$ has been omitted, since it is related to $H_1,H_2$ and the Casimir $\bar K$ by 
\beq\label{kbid} \bar{K}=\alpha^3 {H}_0+\alpha^2 {H}_1+\alpha {H}_2+1.\eeq 
Then a direct computation of the bracket $\{H_1,H_2\}=0$ using (\ref{s5br}) shows that the map (\ref{shat5}) is Liouville 
integrable.   

By setting $\hat{s}_n=u_{2n} u_{2n+1}$, it follows from Theorem  \ref{todatheorem} that the quantities $H_j$ coming from the monodromy matrix pull back to first integrals 
for  (\ref{lift6}). By a slight abuse of notation, we use the same symbols to denote the pullbacks of these integrals, and 
explain how they can be rewritten as functions of the quantities $K_j$ found previously. The key point is that, for fixed $K_0$, 
the spectral curve in the $(\ze , \mu )$ plane  coming from the monodromy matrix, that is 
$$ 
\det ({\cal M}_5 (\ze ) - \mu \mathbf{1} )\equiv \mu^2 - (\ze^3+H_2 \ze^2+H_1\ze+H_0 )\mu + 1-\al \ze =0, 
$$ 
is isomorphic to (\ref{spec5}) via the change of coordinates
$$ 
\ze = \ka -\zet^{-1}, \quad \mu = -\xi \zet^{-3}, \qquad \mathrm{with} \qquad  K_0=1-\al\ka. 
$$ 
This leads to the relations 
$$ 
 H_0 = K_3 -\ka K_2+\ka^2 K_3-\ka^3,
 \qquad H_1=K_2-2\ka K_1+3\ka^2, \qquad H_2=K_1-3\ka, 
$$ 
so that the identity (\ref{kbid}) for $\bar K$ follows immediately from (\ref{kbarid}).

\section{Conclusions} 

We have shown that the key to understanding the integrability of the  family of maps  considered in \cite{DTKQ} is to introduce an additional 
parameter $\beta$, as in (\ref{equ}), and then  lift  to one dimension higher, eliminating this parameter  to obtain (\ref{ushift}). 
Although the properties of the map differ according to the parity of the dimension $N$, 
the Poisson bracket preserved by (\ref{ushift}),  in  $N+1$ dimensions, is given by the same formulae (\ref{nubr}) for both even and odd $N$.   For the case of even $N$, we have found that the Liouville integrability of (\ref{ushift}) follows from 
the corresponding results for reductions of  Hirota's lattice KdV equation, considered in previous work. 
For odd $N$, the situation is more complicated: 
the connection with  a reduction of the bilinear discrete KP (Hirota-Miwa) equation provides  a Poisson bracket, a Lax pair, 
and a set of first integrals, but showing that these are in involution requires more work, and a general proof is lacking. 
On the other hand, for $N=2P-1$, there 
is an intriguing connection with a B\"acklund transformation (BT) for the $(1,-P)$ reduction of the discrete time Toda equation  
(\ref{Todaeq}). For the general $(Q,-P)$ Toda reductions, considered briefly in \cite{hkw},  it would be interesting to construct a BT and see if there is a natural analogue of (\ref{ushift}) for $Q>1$.  

The  starting point for all of the results in section 3 was the derivation of the Hirota bilinear equations associated with (\ref{equ}). This was achieved in two ways: first of all, in section 2, via an experimental approach involving the singularity 
confinement test (or an arithmetical version of it), followed by the lift to a Laurentification of (\ref{equ}), whose tropical 
analogue yields an exact calculation of degree growth; and secondly, once numerical and symbolic calculations  produced  
bilinear equations for particular (small enough) values of $N$, by proving suitable algebraic identities in the general case. 
This combination of analytical, numerical and algebraic methods appears to be very effective, and we propose to apply it to 
other families of difference equations or maps in the future. 

\vspace{.1in}

\noindent 
{\bf Acknowledgements:} 
ANWH is supported by EPSRC fellowship EP/M004333/1. 
This collaboration was supported by the Australian Research Council. We are grateful to Dinh Tran and Peter van der Kamp for useful discussions on related matters.


\begin{thebibliography}{3}

\bibitem{aba} N. Abarenkova, J.-Ch. Angles d'Auriac, S. Boukraa, S. Hassani and
J.-M. Maillard, 
Real Arnold complexity versus real topological entropy for
birational transformations, 
J. Phys. A: Math. Gen. {\bf 33} (2000) 1465--1501.

\bibitem{chang} X.-K. Chang, X.-B. Hu and G. Xin, 
Hankel Determinant Solutions to Several Discrete Integrable Systems and the Laurent Property, 
SIAM J. Discrete Math. {\bf 29} 
(2015) 667--682.


\bibitem{Date}
E. Date, M. Jimbo, T. Miwa, Method for generating discrete soliton equations III, J. Phys. Soc. Japan \textbf{52} 
(1983)  388--393.

\bibitem{DTKQ}
D.K. Demskoi, D.T. Tran, P.H. van der Kamp, and G.R.W. Quispel,   A novel $n$th order difference equation that may be integrable, J. Phys. A: Math. Theor. {\bf{45}} (2012) 135202.  



\bibitem{fz} S. Fomin and A. Zelevinsky, The Laurent Phenomenon, Adv Appl Math {\bf 28} (2002) 119--144. 

\bibitem{f1} S. Fomin and A. Zelevinsky,  Cluster algebras IV: coefficients, Comp. Math. {\bf 143}  (2007) 
112--164. 

\bibitem{f2} S. Fomin, M. Shapiro and D. Thurston, Cluster algebras and triangulated surfaces. Part I:
Cluster complexes, Acta Mathematica {\bf 201} (2008) 83--146.

\bibitem{FH}
A.P. Fordy and A.N.W.Hone, Discrete integrable systems and Poisson algebras from cluster maps, Commun. Math. Phys. \textbf{325} (2014) 527--584.


\bibitem{grp} B. Grammaticos, A. Ramani, V. Papageorgiou, Do integrable mappings have the Painlev\'e property?, 
Phys. Rev. Lett. {\bf 67} (1991) 1825--1828.

\bibitem{con1} B. Grammaticos, A. Ramani, R. Willox, T. Mase and J. Satsuma, 
Singularity confinement and full-deautonomisation: A discrete
integrability criterion, 
Physica D {\bf 313} (2015) 11--25. 



\bibitem{halburd} R.G. Halburd, Diophantine integrability, J. Phys. A: Math. Gen. {\bf 38} (2005) L263--L269. 



\bibitem{hk}  K. Hamad and P.H. van der Kamp, From discrete integrable equations to Laurent recurrences,
J Differ Equ Appl {\bf 22}  (2016) 789--816. 

\bibitem{hhkq} K. Hamad, A.N.W. Hone, P.H. van der Kamp, G.R.W. Quispel, 
QRT maps and related Laurent systems, {\tt arXiv:1702.07047 } 

\bibitem{hamad} K. Hamad, Laurentification, PhD thesis, La Trobe University, 2017.

\bibitem{hv} J. Hietarinta and C.M. Viallet, Singularity Confinement and Chaos in Discrete Systems, Phys. Rev. Lett. {\bf 81} (1998) 325--328.

\bibitem{book}
J. Hietarinta, N. Joshi, F. W. Nijhoff, Discrete Systems and Integrability,  {\it  Cambridge University Press} (2016). 

\bibitem{honetams} A.N.W. Hone, Sigma function solution of the initial value problem for Somos 5 sequences, 
Trans. Amer. Math. Soc.  {\bf 359} 
(2007)  5019--5034. 

\bibitem{honeconf}  A.N.W. Hone, Singularity confinement for maps with the Laurent property, 
Physics Letters A {\bf 361} (2007) 341--345. 

\bibitem{hones6}  A.N.W. Hone, Analytic solutions and integrability for 
bilinear recurrences of order six,
Applicable Analysis {\bf 89}  (2010) 473--492.


\bibitem{hkqt}  A.N.W. Hone, P.H. van der Kamp, G.R.W. Quispel and D.T. Tran, Integrability of
reductions of the discrete Korteweg-de Vries and potential Korteweg-de Vries equations,
Proc. R. Soc. A {\bf 469} (2013) 20120747. 

\bibitem{hi} 
A.N.W. Hone and R. Inoue, Discrete Painlev\'{e} equations from Y-systems, Journal of
Physics A: Mathematical and Theoretical {\bf 47} (2014) 474007. 

\bibitem{hkw} A.N.W. Hone, T.E. Kouloukas and C. Ward, 
On reductions of the Hirota-Miwa equation, Symmetry, Integrability and Geometry: Methods and Applications, 
{\bf 13} (2017) 057. 

\bibitem{kanki1}  M. Kanki, J. Mada, K.M.  Tamizhmani and T. Tokihiro, Discrete Painlev\'{e}  II equation over finite fields,
J. Phys. A: Math. Theor. {\bf 45} (2012) 342001.

\bibitem{kanki2} M.Kanki, 
Integrability of Discrete Equations Modulo a Prime, 
Symmetry, Integrability and Geometry: Methods and Applications 
{\bf 9} (2013) 056.

\bibitem{kuskly} V.B. Kuznetsov and E.K. Sklyanin, 
On B\"{a}cklund transformations for many-body systems, 
J. Phys. A {\bf 31} (1998) 2241--2251. 

\bibitem{kuvan} V.B. Kuznetsov and P. Vanhaecke,
B\"{a}cklund transformations for finite-dimensional integrable systems: a geometric approach, 
J. Geom. Phys. {\bf 44} (2002) 1--40.

\bibitem{maeda}  S. Maeda, Completely integrable symplectic mapping, Proc. Jpn. Acad. {\bf 63}  Ser. A (1987) 198--200.

\bibitem{mq}K. Maruno and G.R.W. Quispel, 
Construction of Integrals of Higher-Order Mappings, 
J. Phys. Soc. Japan {\bf 75}  (2006) 123001. 

\bibitem{mase1} T. Mase, The Laurent Phenomenon and Discrete Integrable
Systems, RIMS K\^{o}ky\^{u}roku Bessatsu
{\bf B41} (2013) 043--064. 

\bibitem{mase2} T. Mase, 
Investigation into the role of the Laurent property in integrability, Journal
of Mathematical Physics 57 (2016) 022703. 

\bibitem{con2} 
T. Mase, R. Willox, B. Grammaticos and A. Ramani, 
Deautonomization by
singularity confinement:
an algebro-geometric
justification, 
Proc. R. Soc. A {\bf 47} (2015) 20140956. 

\bibitem{okubo} N. Okubo, Discrete Integrable Systems and Cluster Algebras, 
 RIMS K\^{o}ky\^{u}roku Bessatsu 
{\bf B41} (2013) 025--041.

\bibitem{svdp} A.J. van der Poorten and C.S. Swart,  Recurrence relations for elliptic sequences: Every Somos 4 is a Somos k,  Bull. Lond. Math. Soc. {\bf 38} (2006) 546--554.

\bibitem{qcr} G.R.W. Quispel, H.W. Capel and J.A.G. Roberts, Duality for discrete integrable systems, J. Phys. A: Math. Gen. {\bf 38} (2005) 3965--3980. 
 
\bibitem{con3} A Ramani, B Grammaticos, R Willox, T Mase and M Kanki, 
The redemption of singularity confinement, 
J. Phys. A: Math. Theor. {\bf  48} (2015) 11FT02. 

\bibitem{rob} R. Robinson, Periodicity of Somos sequences, Proc. Amer. Math. Soc. {\bf 116} (1992) 613--619.

\bibitem{Suris1}
Y.B. Suris, Bi-Hamiltonian structure of the $qd$ algorithm and new discretizations of the Toda lattice {\it{Phys. Lett. A}} {\bf{206}} (1995) 153--161. 

\bibitem{Suris2}
Y.B. Suris, A collection of integrable systems of the Toda type in continuous and discrete time, with $2 \times 2$ Lax pair representations (1997),  
{\tt arXiv:solv-int/9703004v2}



\bibitem{veselov}  A.P. Veselov, Integrable maps, Russ. Math. Surv. {\bf 46}  (1991) 1--51. 

\bibitem{viallet} C.M. Viallet, On the algebraic structure of rational discrete dynamical systems, J Phys A:
Math Theor {\bf 48} (2015) 16FT01.



\end{thebibliography}
\end{document}